\documentclass[twocolumn]{revtex4-1}
\usepackage{amssymb, amsfonts, amsmath, amsthm, amsthm, bbm, changes}
\pdfoutput=1
\usepackage{url}
\usepackage{hyperref}
\def\unity{\mathbbm 1}
\def\tr{\mathrm{tr}}
\def\U{\mathrm{U}}
\def\SU{\mathrm{SU}}
\def\A{\mathsf A}
\def\B{\mathsf B}
\def\C{\mathsf C}
\def\E{\mathsf E}
\def\D{\mathcal N}
\def\M{\mathcal M}
\def\V{\mathcal V}

\def\F{\mathcal F}

\def\u{\mathbf u}
\def\f{\mathbf f}
\def\g{\mathbf g}
\def\h{\mathbf h}
\def\complex{\mathbb C}

\newcommand{\proj}[1]{|#1\rangle\! \langle#1|}

\newtheorem{Definition}{Definition}
\newtheorem{Lemma}[Definition]{Lemma}
\newtheorem{Theorem}[Definition]{Theorem}
\newtheorem{Corollary}[Definition]{Corollary}

\begin{document}

\title{The measurement postulates of quantum mechanics are operationally redundant}

\author{Llu\'\i s Masanes}
\affiliation{Department of Physics and Astronomy, University College London, Gower Street, London WC1E 6BT, United Kingdom}
\author{Thomas D.~Galley} 
\email{tgalley1@perimeterinstitute.ca}
\affiliation{Department of Physics and Astronomy, University College London, Gower Street, London WC1E 6BT, United Kingdom}
\affiliation{Perimeter Institute for Theoretical Physics, Waterloo, ON N2L 2Y5, Canada}
\author{Markus P.~M\"uller}
\affiliation{Institute for Quantum Optics and Quantum Information, Austrian Academy of Sciences, Boltzmanngasse 3, A-1090 Vienna, Austria}
\affiliation{Perimeter Institute for Theoretical Physics, Waterloo, ON N2L 2Y5, Canada}

\date{\today}

\begin{abstract}\noindent
Understanding the core content of quantum mechanics requires us to disentangle the hidden logical relationships between the postulates of this theory.
Here we show that the mathematical structure of quantum measurements, the formula for assigning outcome probabilities (Born's rule) and the post-measurement state-update rule, can be deduced from the other quantum postulates, often referred to as ``unitary quantum mechanics", and the assumption that ensembles on finite-dimensional Hilbert spaces are characterised by finitely many parameters.
This is achieved by taking an operational approach to physical theories, and using the fact that the manner in which a physical system is partitioned into subsystems is a subjective choice of the observer, and hence should not affect the predictions of the theory. 
In contrast to other approaches, our result does not assume that measurements are related to operators or bases, it does not rely on the universality of quantum mechanics, and it is independent of the interpretation of probability.
\end{abstract}
\maketitle

\section{Introduction}

\noindent
What sometimes is postulated as a fundamental law of physics is later on  understood as a consequence of more fundamental principles.
An example of this historical pattern is the rebranding of the symmetrization postulate as the spin-statistics theorem \cite{PhysRev.82.914}.
Another example, according to some authors, is the Born rule, the formula that assigns probabilities to quantum measurements. 
The Born rule has been derived within the framework of quantum logic \cite{Gleason_measures_1957, Cooke_elementary_1985, Pitowsky_infinite_1998, Wilce_quantum_2017}, taking an operational approach \cite{Saunders_derivation_2004, Busch_quantum_2003, Caves_Gleason_2004,wright_gleason-type_2018}, and using other methods \cite{Logiurato_born_2012, auffeves_simple_2015, Han_Quantum_2016,Frauchiger_nonprobabilistic_2017,shrapnel_updating_2018}.
But all these derivations assume, among other things, the mathematical structure of quantum measurements, that is, the correspondence between measurements and orthonormal bases, or more generally, positive-operator valued measures \cite{Holevo_statistical_1973, Helstrom76}.

Taking one step further, the structure of measurements together with the Born rule can be jointly derived within the many-worlds interpretation of quantum mechanics (QM) \cite{Deutsch_quantum_1999, Wallace_how_2010} and the framework of entanglement-assisted invariance \cite{Zurek_probabilities_2005, Zurek20180107,Barnum_no_2003,Schlosshauer_zureks_2005}. But these derivations involve controversial uses of probability in deterministic multiverse scenarios, which have been criticized by a number of authors \cite{Barnum_quantum_2000, Kent_one_2010, Baker_measurement_2006, Hemmo_quantum_2007, Lewis_Peter_2010, Price_decisions_2010, Albert_probability_2010, Caves_note_2004, Schlosshauer_zureks_2005, Barnum_no_2003, mohrhoff_probabilities_2004}.
Also, these frameworks require the universality of QM, meaning that the measurement apparatus and/or the observer has to be included in the quantum description of the measuring process.
While this is a meaningful assumption, it is interesting to see that it is not necessary, as proven in the present article.

In this work we take an operational approach, with the notions of measurement and outcome probability being primitive elements of the theory, but without imposing any particular structure on them.
We use the fact that the subjective choices in the description of a physical setup in terms of operational primitives must not affect the predictions of the theory.
For example, deciding to describe a tripartite system $\A\cdot \B\cdot \C$ as either the bipartite system $\A\B\cdot \C$ or as $\A\cdot \B\C$ must not modify the outcome probabilities.
Using these constraints we characterize all possible alternatives to the mathematical structure of quantum measurements and the Born rule, and we prove that there is no such alternative to the standard measurement postulates.
This theorem has simple and precise premises, it does not require unconventional uses of probability theory, and it is independent of the interpretation of probability.
A further interesting consequence of this theorem is that the post-measurement state-update rule must necessarily be that of QM.

The structure of this article is the following. Section~\ref{sec:main result} reviews the postulates of QM, introduces a new formalism that allows to specify any alternative to the measurement postulates, and uses this formalism to state the main result of this work: the measurement theorem.
Section~\ref{sec:discussion} illustrates this theorem with two interesting examples, and contrasts our result with Gleason's theorem \cite{Gleason_measures_1957}.
Section~\ref{sec:methods} provides a bird's eye view of the proof of the theorem, which is fully detailed in the appendicies.
Finally, Section~\ref{sec:conclusions} concludes with some important remarks.

\section{Results}\label{sec:main result}

\noindent

\subsection{The standard postulates of QM}

Before presenting the main result we prepare the stage appropriately.
This involves reviewing some of the postulates of QM, reconstructing the structure of mixed states from them, and introducing a general characterization of measurements that is independent of their mathematical structure.

\medskip\noindent \textbf{Postulate (states).}
\emph{To every physical system there corresponds a complex and separable Hilbert space $\complex^d$, and the pure states of the system are the rays $\psi \in \mathrm P\complex^d$.}

\medskip\noindent
It will be convenient to use the notation $\complex^d$ both for Hilbert spaces of finite dimension $d$, and also for countably infinite-dimensional Hilbert spaces which we denote by $\complex^\infty$. This notation is justified, since all countably infinite-dimensional Hilbert spaces are isomorphic~\cite{rudin1991functional}.
Analogously we use $U(\infty)$ to denote the unitary transformations of $\complex^\infty$. In this document we represent states (rays) by normalized vectors $\psi \in \complex^d$.

\medskip\noindent \textbf{Postulate (transformations).}
\emph{The reversible transformations (for example, possible time evolutions) of pure states of $\complex^d$ are the unitary transformations $\psi \mapsto U\psi$ with $U\in \U(d)$.}

\medskip\noindent \textbf{Postulate (composite systems).}
\emph{The joint pure states of systems $\complex^a$ and $\complex^b$ are the rays of the tensor-product Hilbert space  $\complex^a \otimes \complex^b$.}

\medskip\noindent \textbf{Postulate (measurement).}
\emph{Each measurement outcome of system $\complex^d$ is represented by a linear operator $Q$ on $\complex^d$ satisfying $0\leq Q \leq \unity$, where $\unity$ is the identity.}
The probability of outcome $Q$ on state $\psi \in \complex^d$ is 
\begin{equation}
  \label{def:Born Rule}
  P(Q|\psi)
  = 
  \langle\psi| Q |\psi\rangle\ .
\end{equation}
A (full) measurement is represented by the operators corresponding to its outcomes $Q_1,\ldots, Q_n$, which must satisfy the normalization condition $\sum_{i=1}^n Q_i = \unity$.

\medskip\noindent
The more traditional formulation of the measurement postulate in terms of (not necessarily positive) Hermitian operators is equivalent to the above. 
But we have chosen the above form because it is closer to the formalism used in the presentation of our results.

\medskip\noindent \textbf{Postulate (post-measurement state-update).}
\emph{Each outcome is represented by a completely-positive linear map $\Lambda$ related to the operator $Q$ via
\begin{equation}
  \tr \Lambda (\proj\psi)
  =
  \langle\psi| Q |\psi\rangle
  \ ,
\end{equation}
for all $\psi$.
The post-measurement state after outcome $\Lambda$ is
\begin{equation}
  \rho
  =
  \frac {\Lambda (\proj\psi)}
  {\tr \Lambda (\proj\psi)}\ .
\end{equation}
A (full) measurement is represented by the maps corresponding to its outcomes $\Lambda_1, \ldots, \Lambda_n$ whose sum $\sum_{i=1}^n \Lambda_i$ is trace-preserving.}

\medskip
If the measurement is repeatable and minimally disturbing~\cite{Ozawa_quantum_1984, Chiribella_Sharpness_2014} then $Q_1,\ldots, Q_n$ are projectors and the above maps are of the form $\Lambda_i (\rho) = Q_i \rho Q_i$, which is the standard textbook ``projection postulate". 
Below we prove that the ``measurement" and ``post-measurement state-update" postulates are a consequence of the first three postulates.

\subsection{The structure of mixed states}
\label{SubsecMixed}

\noindent
Mixed states are not mentioned in the standard postulates of QM, but their structure follows straightaway from the measurement postulate \eqref{def:Born Rule}.
Recall that a mixed state is an equivalence class of indistinguishable ensembles, and an ensemble $(\psi_r, p_r)$ is a probability distribution over pure states. 
Note that the notion of distinguishability depends on what the measurements are.
For the particular case of quantum measurements \eqref{def:Born Rule}, the probability of outcome $Q$ when a source prepares state $\psi_r$ with probability $p_r$ is
\begin{equation}\label{eq:ensemb prob}
  P\big(Q\big|(\psi_r, p_r) \big)
  =
  \sum_r p_r P\big(Q\big|\psi_r \big)
  =
  \tr\big( Q \rho \big)\ ,
\end{equation}
where we define the density matrix
\begin{equation}
  \label{def:density matrix}
  \rho = \sum_r p_r \proj{\psi_r}\ .
\end{equation}
This matrix contains all the statistical information of the ensemble.
Therefore, two ensembles with the same density matrix are indistinguishable.

The important message from the above is that a different measurement postulate would give different equivalence classes of ensembles, and hence, a different set of mixed states. Thus, in proving our main result, we will not assume that mixed states are of the form~(\ref{def:density matrix}).
An example of mixed states for a non-quantum measurement postulate is described in section~\ref{sec:example 1}.

\subsection{Formalism for any alternative measurement postulate}

\noindent
Before proving that the only possible measurement postulate is that of QM, we have to articulate what ``a measurement postulate" is in general.
In order to do so, we introduce a theory-independent characterization of measurements for single and multipartite systems.
This is based on the concept of outcome probability function (OPF), introduced in \cite{Galley_classification_2017} and defined next.

\medskip\noindent \textbf{Definition (OPF).}
\emph{Each measurement outcome that can be observed on system $\complex^d$ is represented by the function $\f: \mathrm P \complex^d \to [0,1]$ being its corresponding probability $\f(\psi) = P(\f|\psi)$ for each pure state $\psi \in \mathrm P \complex^d$;
and we denote by $\F_d$ the complete set of OPFs of system $\complex^d$. }
Completeness is defined below as the closure of $\F_d$ under various operations.

\medskip\noindent
If instead of a single outcome we want to specify a full measurement with, say, $n$ outcomes, we provide the OPFs $\f_1,\ldots ,\f_n$ corresponding to each outcome; which must satisfy the normalization condition
\begin{equation}
  \label{normalization condition}
  \sum_{i=1}^n \f_i(\psi) =1\ ,
\end{equation}
for all states $\psi$.

It is important to note that this mathematical description of measurements is independent of the underlying interpretation of probability: all we are assuming is that there exist experiments which yield definite outcomes (possibly relative to a given agent who uses this formalism), and that it makes sense to assign probabilities to these outcomes. For example, we could interpret them as Bayesian probabilities of a physicist who bets on future outcomes of experiments; or as limiting frequencies of a large number of repetitions of the same experiment, approximating empirical data. Whenever we have an experiment of that kind, the corresponding probabilities (whatever they mean) will be determined by a collection of OPFs.

The completeness of the set of OPFs $\F_d$ consists of the following three properties: 

\emph{$\F_d$ is closed under taking mixtures.}
Suppose that the random variable $x$ with probability $p_x$ determines which 2-outcome measurement $\f_1^x, \f_2^x \in \F_d$ we implement, and later on we forget the value of $x$. 
Then the probability of outcome 1 for this ``averaged" measurement is
\begin{equation}
  \label{closed mixtures}
  \sum_x p_x\, \f_1^x \in \F_d\ ,
\end{equation}
which must be a valid OPF.
Therefore, mixtures of OPFs are OPFs.

\emph{$\F_d$ is closed under composition with unitaries.}
We can always perform a transformation $U\in \U(d)$ before a measurement $\f \in \F_d$, effectively implementing the measurement 
\begin{equation}
  \label{SU action on F}
  \f\circ U \in \F_d\ ,
\end{equation}
which then must be a valid OPF.
Note that here we are not saying that all unitaries can be physically implemented, but only that the formalism must in principle include them.

\emph{$\F_d$ is closed under systems composition.}
Since $\F_d$ is complete, it also includes the measurements that appear in the description of $\complex^d$ as part of the larger system $\complex^d \otimes \complex^b \cong \complex^{db}$, for any background system $\complex^b$. 
Formally, for each background state $\varphi\in \complex^b$ and global OPF $\g \in \F_{db}$ there is local OPF $\f_{\varphi,\g} \in \F_d$ which represents the same measurement outcome
\begin{align}
  \label{closed with ancilla}
  \f_{\varphi,\g} (\psi)
  =
  \g(\psi \otimes \varphi) \ , 
\end{align}
for all $\psi \in \mathrm P\complex^d$.

Next we consider local measurements in multipartite systems.
In order to do so, it is useful to recall that the observer always has the option of describing a systems $\complex^a$ as part of a larger system $\complex^a \otimes \complex^b$, without this affecting the predictions of the theory. 
In order to do so, the observer needs to know how to represent the OPFs of the small system $\F_a$ as OPFs of the larger system $\F_{ab}$.
This information is contained in the star product, defined in what follows.

\medskip\noindent \textbf{Definition ($\star$-product).}
\emph{Any pair of local OPFs, $\f \in \F_a$ and $\g\in \F_b$, is represented as a global OPF $(\f\star\g) \in \F_{ab}$ via the star product $\star: \F_a \times \F_b \rightarrow \F_{ab}$, which satisfies
\begin{equation}
  \label{prod for prod}
  (\f \star \g)(\psi \otimes \varphi)
  = 
  \f (\psi)\, \g (\varphi)\ ,
\end{equation}
for all $\psi\in \mathrm P\complex^a$ and $\varphi\in \mathrm P\complex^b$.
This product must be defined for any pair of (complex and separable) Hilbert spaces $\complex^a$ and $\complex^b$.}

\medskip\noindent
In other words, the $\star$-product represents bi-local measurements, which in QM are represented by the tensor product in the space of Hermitian matrices.

Since the option of describing system $\complex^a$ as part of a larger system $\complex^a \otimes \complex^b$ is a subjective choice that must not affect the predictions of the theory, the embedding of $\F_a$ into $\F_{ab}$ provided by the $\star$-product must preserve the structure of $\F_a$.
This includes the mixing (convex) structure
\begin{equation}
  \label{commut with mixtures}
  \left(\mbox{$\sum_x$}\, p_x\, \f^x\right) \star \g
  =
  \mbox{$\sum_x$}\, p_x \left(\f^x\! \star \g \right)\ ,
\end{equation}
as well as the $\U(d)$ action
\begin{equation}
  \left(\f \circ U\right) \star \g
  =
  \left(\f \star \g \right) \circ (U\otimes \unity_b)\ .
\end{equation}
And likewise for the other party $\F_b$.
The $\star$-product must also preserve probability, in the sense that if $\{\f_i\} \subseteq \F_a$ and $\{\g_j\} \subseteq \F_b$ are full measurements satisfying the normalization condition \eqref{normalization condition} then we must have
\begin{equation}
    \label{cond:u*u=u}
    \left[(\mbox{$\sum_{i}$} \f_i)
    \star
    (\mbox{$\sum_{j}$} \g_i)\right]\!
    (\psi) =1\ ,
\end{equation}
for all rays $\psi$ of $\complex^a \otimes \complex^b$.

Pushing the same philosophy further, the observer has the option of describing the tripartite system $\complex^a \otimes \complex^b \otimes \complex^c$ as the bipartite system $\complex^a \otimes [\complex^b \otimes \complex^c]$ or the bipartite system $[\complex^a \otimes \complex^b] \otimes \complex^c$, without this affecting the probabilities predicted by the theory.
This translates to the $\star$-product being associative
\begin{equation}
  \label{cond:associativity}
  \f \star \left(\g \star  \h \right)
  =
  (\f \star \g )\star  \h\ .
\end{equation}
That is, the probability of outcome $\f\star \g\star  \h$ is independent of how we choose to partition the global system into subsystems. As we show below, this property will be crucial to recover the standard measurement postulates of quantum mechanics.

\subsection{The measurement theorem}

\noindent
Before stating the main result of this work, we specify what should be the content of any alternative measurement postulate, and state an operationally-meaningful assumption that is necessary to prove our theorem.

\medskip\noindent \textbf{Definition (measurement postulate).}
\emph{This is a family of OPF sets $\F_2, \F_3, \F_4, \ldots$ and $\F_\infty$ equipped with a $\star$-product $\F_a \times \F_b \to \F_{ab}$ satisfying conditions (\ref{closed mixtures}-\ref{cond:associativity}).}

\medskip\noindent
In addition to the above, a measurement postulate could provide restrictions on which OPFs can be part of the same measurement (beyond the normalization condition).
However, such rules would not affect our results. 


\medskip\noindent \textbf{Assumption (possibility of state estimation).}
\emph{Each finite-dimensional system $\complex^d$ has a finite list of outcomes $\f^1, \ldots, \f^k \in \F_d$ such that knowing their value on any ensemble $(\psi_r, p_r)$ allows us to determine the value of any other OPF $\g\in \F_d$ on the ensemble $(\psi_r, p_r)$.}

\medskip\noindent
It is important to emphasize that $\f^1, \ldots, \f^k$ need not be outcomes of the same measurement; and also, this list need not be unique.
For example, in the case of QM, we can specify the state of a spin-$\frac 1 2$ particle with the probabilities of outcome ``up'' in any three linearly independent directions. Also in QM, we have $k=d^2-1$; but here we are not assuming any particular relation between $d$ and $k$.
Now it is time to state the main result of this work, which essentially tells us that the only possible measurement postulates are the quantum ones.

\medskip\noindent \textbf{Theorem (measurement).}
\emph{The only measurement postulate satisfying the ``possibility of state estimation" has OPFs and $\star$-product of the form
\begin{align}
  &\f(\varphi) = 
  \langle\varphi| F|\varphi\rangle\ ,  
  \label{eqFormOfOPFs}
  \\
  &(\f\star\g)(\psi) =
  \langle\psi| F\otimes G |\psi\rangle\ ,
  \label{eqFormOfstar}
\end{align}
for all $\varphi\in \complex^a$ and $\psi \in \complex^a \otimes \complex^b$, where the $\complex^a$-operator $F$ satisfies $0\leq F \leq \unity$, and analogously for $G$.}

\medskip\noindent
The methods section provides a summary of the ideas and techniques used in the proof of this theorem. Full detail can be found in Appendix~\ref{app:nonassoc} and Appendix~\ref{app:inftdim spaces}.

\subsection{The post-measurement state-update rule}

\noindent
At first sight, the above theorem says nothing about the post-measurement state-update rule. 
But actually, it is well-known~\cite{DaviesLewis} that the only possible state-update rule that is compatible with the probability rule implied by the theorem (\ref{eqFormOfOPFs}-\ref{eqFormOfstar})
is the one stated above in postulate ``post-measurement state-update rule". 
We include a self-contained proof of the above in Appendix~\ref{app:post-measurement}. 


\section{Discussion}\label{sec:discussion}

\subsection{Non-quantum measurement postulate violating associativity}
\label{sec:example 1}

\noindent
In this section we present an example of alternative measurement postulate, which shows that it is possible to bypass the measurement theorem if we give up the associativity condition \eqref{cond:associativity}. 
It also illustrates how a different choice of measurement postulate produces a different set of mixed states.

\medskip\noindent \textbf{Definition (non-quantum measurement postulate).}
\emph{An $n$-outcome measurement on $\complex^a$ is characterized by $n$ Hermitian operators $F_i$ acting on $\complex^a \otimes \complex^a$ and satisfying $0\leq F_i \leq P_+^a$ and
\begin{equation}
  \sum_{i=1}^n F_i
  = P_+^a \ ,
\end{equation}
where $P_+^a$ is the projector onto the symmetric subspace of $\complex^a \otimes \complex^a$.
The probability of outcome $i$ on the (normalized) state $\varphi  \in \complex^a$ is given by
\begin{equation}
  \label{ex:opf}
  \f_i (\varphi) = \tr\!\left(
  F_i \proj\varphi ^{\otimes 2} \right)\ ;
\end{equation}
and the $\star$-product of two OPFs $\f\in \F_a$ and $\g\in \F_b$ of the form \eqref{ex:opf} is defined as
\begin{equation*}
  (\f \star \g)(\psi) 
  = \tr\!\left[
  \left( F\otimes G + 
  \mbox{$\frac{\tr\, F}{\tr P^a_+}$} P^a_-
  \otimes 
  \mbox{$\frac{\tr\, G}{\tr P^b_+}$} P^b_-
  \right)\! \proj{\psi}^{\otimes 2} \right],
\end{equation*}
for any normalized $\psi \in \complex^a \otimes \complex^b$.}

\medskip\noindent
This alternative theory violates the principles of ``local tomography"~\cite{Hardy_quantum_2001} and ``purification" \cite{Chiribella_probabilistic_2010}.
This and other exotic properties of this theory are analyzed in detail in previous work \cite{Galley_classification_2017,Galley_impossibility_2018}.
Also, the validity of marginal and conditional states imposes additional constraints on the matrices $F$ which are also worked out in \cite{Galley_impossibility_2018}.
It is easy to check that the above definition satisfies conditions (\ref{closed mixtures}-\ref{cond:u*u=u}) and violates associativity \eqref{cond:associativity}.
Therefore, this provides a perfectly valid toy theory of systems that encompass either one or two components, but not more.

As we have mentioned above, the structure of the mixed states depends on the measurement postulate.
Here, the mixed state corresponding to ensemble $(\psi_r, p_r)$ is
\begin{equation}
  \omega =
  \sum_r p_r \proj{\psi_r}^{\otimes 2}\ .
\end{equation}
Another non-quantum property of this toy theory is that the uniform ensembles corresponding to two different orthonormal bases, $\{ \varphi_i \}$ and $\{\psi_i \}$ are distinguishable
\begin{align}
  \sum_i \frac 1 d \,
  \proj{\varphi_i}^{\otimes 2}
  \neq
  \sum_i \frac 1 d \,
  \proj{\psi_i}^{\otimes 2}\ .
\end{align}



\subsection{Gleason's theorem and non-contextuality}

\noindent
As mentioned in the introduction, Gleason's theorem and many other derivations of the Born rule \cite{Gleason_measures_1957, Cooke_elementary_1985, Pitowsky_infinite_1998, Wilce_quantum_2017, Saunders_derivation_2004, Busch_quantum_2003, Caves_Gleason_2004, Logiurato_born_2012, Han_Quantum_2016} assume the structure of quantum measurements; that is, the correspondence between measurements and orthonormal bases $\{ \varphi_i \}$, or more generally, positive-operator valued measures \cite{Helstrom76}.
But in addition to this, they assume that the probability of an outcome $\varphi_i$ does not depend on the measurement (basis) it belongs to.
Note that this type of ``non-contextuality"  is already part of the content of Born's rule.

To show that this ``non-contextuality"  assumption is by no means necessary, we review an alternative to the Born rule, presented in \cite{Aaronson_quantum_2004}, which does not satisfy it.
In this toy theory, we also have that measurements are associated to orthonormal bases $\{ \varphi_i \}$ and each outcome corresponds to an element $\varphi_i$ of the basis. 
Then, the probability of outcome $\varphi_i$ on state $\psi$ is given by
\begin{equation}
  \label{prob ex2}
  P(\varphi_i|\psi) = \frac
  {|\langle \varphi_i|\psi\rangle|^{4}}
  {\sum_j|\langle \varphi_j|\psi\rangle|^{4}}\ .
\end{equation}
Since this example does not meet the premises of Gleason's theorem (the denominator depends not only on $\varphi_i$ but also on the rest of the basis), there is no contradiction in that it violates its conclusion.

We stress that our results, unlike previous contributions \cite{Gleason_measures_1957, Cooke_elementary_1985, Pitowsky_infinite_1998, Wilce_quantum_2017, Saunders_derivation_2004, Busch_quantum_2003, Caves_Gleason_2004, Logiurato_born_2012, Han_Quantum_2016}, do not assume this type of non-contextuality. In particular, our OPF framework perfectly accommodates the above example~(\ref{prob ex2}) with $\f_i(\psi)=P(\varphi_i|\psi)$. This example however does not meet the ``possibility of state estimation" assumption, and hence is excluded by the main theorem of this paper. 

In Appendix~\ref{app:Other_work} we discuss publications~\cite{Frauchiger_nonprobabilistic_2017} and~\cite{Cabello_the_2018} in relation to the theorem presented in this paper.


\section{Methods}\label{sec:methods}

\noindent
This brief section provides a bird's eye view of the proof of the measurement theorem.
The argument starts by embedding the OPF set $\F_d$ into a complex vector space so that physical mixtures \eqref{closed mixtures} can be represented by certain linear combinations.
Second, the ``possibility of state estimation" assumption implies that, whenever $d$ is finite, this embedding vector space is finite-dimensional.
This translates the $\U (d)$ action \eqref{SU action on F} on the set $\F_d$ to a linear representation; and once in the land of $\U(d)$ representations we have a good map of the territory.

Third, the fact that the argument of the functions in $\F_d$ is a ray (not a vector) imposes a strong restriction to the above-mentioned $\U(d)$ representation.
All these restricted representations were classified by some of the authors in \cite{Galley_classification_2017}. 
This amounts to a classification of all alternatives to the measurement postulate for single systems, that is, when the consistency constraints related to composite systems (\ref{closed with ancilla}-\ref{cond:associativity}) are ignored.
The next steps take composition into account.

Fourth, ``closedness under system composition" \eqref{closed with ancilla} implies that all OPFs $\f\in \F_d$ are of the form 
\begin{equation}
  \label{eq:opf}
  \f (\varphi) = \tr\!\left(
  F \proj\varphi ^{\otimes n} \right)\ ,
\end{equation}
where $n$ is a fixed positive integer.
Recall that the case $n=1$ is QM and the case $n=2$ has been studied above.
In the final step, the representation theory of the unitary group is exploited to prove that, whenever $n\geq 2$, it is impossible to define a star product of functions \eqref{eq:opf} satisfying associativity \eqref{cond:associativity}.
This implies that only the quantum case ($n=1$) fulfils all the required constraints (\ref{closed mixtures}-\ref{cond:associativity}).

\section{Conclusions}\label{sec:conclusions}

\noindent
It may seem that conditions (\ref{closed mixtures}-\ref{cond:associativity}) are a lot of assumptions to claim that we derive the measurement postulates from the non-measurement ones.

But from the operational point of view, these conditions constitute the very definition of measurement, single and multi-partite physical system.
In other words, specifying what  we mean by ``measurement" is in a different category than stating that measurements are characterized by operators acting on a Hilbert space.
Analogously, the rules of probability calculus or the axioms of the real numbers are not explicitly included in the postulates of quantum mechanics.

Note that our results also apply to indistinguishable particles (bosons and fermions), as long as we interpret the tensor product not as a composition of particles, but of the corresponding modes.

It is rather remarkable that none of the three measurement postulates (structure, probabilities and state-update) can be modified without having to redesign the whole theory.
In particular, the probability rule is deeply ingrained in the main structures of the theory.
This fact shows that one need not appeal to any supplementary principles beyond operational primitives to derive the Born rule, nor do we need to make any assumptions about the structure of measurements, unlike previous work~\cite{Saunders_derivation_2004, Aaronson_quantum_2004, Zurek_probabilities_2005,Logiurato_born_2012,Wallace_how_2010,Han_Quantum_2016}.
Finally, having cleared up  unnecessary postulates in the formulation of quantum mechanics, we find ourselves closer to its core message.

\section{Acknowledgements}

\noindent
We are grateful to Jonathan Barrett and Robin Lorenz for discussions about the toy theory of section~\ref{sec:example 1}, which was independently studied by them.
LM acknowledges financial support by the Engineering and Physical Sciences Research Council [grant number EP/R012393/1]. TG acknowledges support by the Engineering and Physical Sciences Research Council [grant number EP/L015242/1]. 
This research was supported in part by Perimeter Institute for Theoretical Physics; research at Perimeter Institute is supported by the Government of Canada through Innovation, Science and Economic Development Canada and by the Province of Ontario through the Ministry of Research, Innovation and Science.
This publication was made possible through the support of a grant from the John Templeton Foundation; the opinions expressed in this publication are those of the authors and do not necessarily reflect the views of the John Templeton Foundation.

%
\bibliographystyle{ieeetr}
\bibliography{refs}{}

\appendix

\section{Alternative measurement postulates for single systems}
\label{app:OPFsingle}

\noindent
In this section we classify all alternative measurement postulates for the case of finite-dimensional single systems, that is, when the constraints associated to the composition of systems (the star product) are ignored.
These results build up on the previous work \cite{Galley_classification_2017} by two of us.

In this section we only consider finite-dimensional Hilbert spaces $\complex^d$ with $d$ a positive integer.
In this case we have $\U (d) \cong \SU(d) \times \U(1)$; and since $\U(1)$ has a trivial action on rays, we only consider $\SU(d)$.
Later on, when addressing the infinite-dimensional case $d= \infty$, we will work with $\U(d)$, since the condition $\det U=1$ is not well-defined when $d= \infty$.

\subsection{Structure of measurements and mixed states}

\begin{Definition}\label{def:F_d}
	$\F_d$ is a set of functions $\f: \mathrm P \complex^d \to [0,1]$ which is closed under composition with unitaries $U\in\SU(d)$
	\begin{align}
		\label{SUd action F}
		U: \f \mapsto (\f\circ U)\ ,
	\end{align}
	closed under convex combinations
	\begin{equation}
		\label{eq:mixing op}
		\sum_x p_x\, \f^x \in\F_d\ ,
	\end{equation}
	and that contains the unit and the zero functions, respectively $\u(\psi)=1$ and $\mathbf 0(\psi) =0$, for all $\psi\in \mathrm P \complex^d$.
\end{Definition}

The unit function $\u$ represents an outcome that happens with probability one.
For example, such unit-probability outcome can be the event corresponding to all outcomes of the measurement $\{\f_i\}$, which by normalization satisfy
\begin{align}
	\sum_i \f_i = \u\ .
\end{align}
Analogously, the zero function $\mathbf 0$ represents a formal outcome that has zero probability irrespectively of the state.

For what comes below, it is convenient to consider the set $\F_d$ as embedded in the complex vector space $\complex \F_d$ generated by itself.
The fact that the group action \eqref{SUd action F} commutes with the mixing operation \eqref{eq:mixing op}
\begin{equation}
	\left(\mbox{$\sum_x$}\, p_x\, \f^x \right)
	\circ U
	=
	\mbox{$\sum_x$}\, p_x (\f^x \circ U)\ ,
\end{equation}
can be extended to arbitrary linear combinations in $\complex \F_d$, providing a complex, linear representation of $\SU(d)$.
While only the elements of $\F_d$ are outcome probability functions (OPFs), any element of $\complex \F_d$ can be interpreted as the expectation value of an observable with complex outcome labels, in analogy to the algebra of observables in QM.
While in QM the space $\complex \F_d$ has dimension $d^2$, here we leave the dimension unconstrained.
However, in what follows, we show that the ``possibility of state estimation" assumption implies that the linear space $\complex \F_d$ is finite-dimensional.
But before this, we recall that the probability of outcome $\f\in\F_d$ on an ensemble $(\psi_r, p_r)$ is given by
\begin{equation}
	\f\!\left[(\psi_r, p_r) \right]
	=
	\sum_r p_r\, \f[\psi_r]\ .
\end{equation}
The above follows from the rules of probability calculus.

\begin{Lemma}\label{lemma:finite k}
	Suppose that the values of the outcomes $\f^1, \ldots, \f^k \in \F_d$ on any given ensemble $(\psi_r, p_r)$ determine the value of any other outcome $\g\in \F_d$ on that ensemble $(\psi_r, p_r)$.
	Then the functions $\{\f^1, \ldots, \f^k, \u\}$ span the linear space $\complex \F_d$.
\end{Lemma}

In other words, knowing the numbers $\f^1[(\psi_r,p_r)]$, \ldots, $\f^k[(\psi_r,p_r)]$ allows us to determine the number $\g[(\psi_r,p_r)]$ without knowing the ensemble $(\psi_r,p_r)$. That is, the latter is some function of the former.

\begin{proof}
	Once the OPFs $\f^1, \ldots, \f^k \in \F_d$ are given we can define the convex set
	\begin{equation}
		\mathcal S_d 
		= 
		\mbox{conv}\! \left\{ 
		\left[ \f^1(\psi), \ldots, \f^k(\psi) \right] 
		\,\,|\,\, \psi\in\mathrm P \complex^d 
		\right\}
		\subseteq \mathbb R^k\ .
	\end{equation}
	Next we note that, the fact that the values of $\f^1, \ldots, \f^k$ determine the value of $\g$ on any ensemble $(\psi_r, p_r)$ means that there is a function $\xi_\g: \mathcal S_d \to [0,1]$ such that
	\begin{equation}
		\sum_r p_r \g(\psi_r) 
		= 
		\xi_\g \!\left[
		\mbox{$\sum_r$} p_r \f^1 (\psi_r), 
		\ldots, 
		\mbox{$\sum_r$} p_r \f^k (\psi_r)
		\right]\ .
	\end{equation}
	Since the above equality holds for all ensembles, it also holds for the pure states $\psi_r$
	\begin{equation}
		\g(\psi_r) 
		= 
		\xi_\g \!\left[
		\f^1 (\psi_r), \ldots, 
		\f^k (\psi_r) \right]\ ,
	\end{equation}
	for all $r$.
	Hence we have
	\begin{align}
		\nonumber
		& \xi_\g \!\left[
		\mbox{$\sum_r$} p_r \f^1 (\psi_r), 
		\ldots, 
		\mbox{$\sum_r$} p_r \f^k (\psi_r)
		\right]
		\\ &\ \ \ \, = 
		\mbox{$\sum_r$}\, p_r\, 
		\xi_\g \!\left[
		\f^1 (\psi_r), \ldots, \f^k (\psi_r)
		\right]\ .
		\label{eqConvexSpecial}
	\end{align}
	It follows that~(\ref{eqConvexSpecial}) also holds true if every appearance of $\psi_r$ is replaced by some ensemble $(\psi_s^{(r)},q_s^{(r)})$, where $s$ labels the possible states and their probabilities. Since $(\f^1[(\psi_s^{(r)},q_s^{(r)})],\ldots,\f^k[(\psi_s^{(r)},q_s^{(r)})])$ can take all values in $\mathcal{S}_d$ by choosing the states and probabilities in a suitable way, this shows that $\xi_\g$ is convex on the full set $\mathcal{S}_d$. This implies that $\xi_\g$ can be affinely extended to all of $\mathbb{R}^k$, i.e.\ there is an affine function $\xi'_\g : \mathbb R^k \to \mathbb R$ which coincides with the previous function $\xi'_\g = \xi_\g$ inside the convex set $\mathcal S_d$.
	The affine nature of the function means
	\begin{equation}
		\xi'_\g(\mbox{$\sum_r$} c_r\, \vec x_r)
		=
		\mbox{$\sum_r$} c_r\, \xi'_\g(\vec x_r)\ ,
	\end{equation}
	for any $c_r\in \mathbb R$ with $\sum_r c_r = 1$ and $\vec x_r \in \mathbb R^k$, but the coefficients $c_r$ are not necessarily positive.
	Any affine function $\xi'_\g : \mathbb R^k \to \mathbb R$ can be written as
	\begin{equation}
		\xi'_\g (\vec x) = 
		\vec e_\g \cdot \vec x + c_\g\ ,
	\end{equation}
	where $\vec e_\g \in \mathbb R^k$ and $c_\g\in \mathbb R$.
	Therefore we can write
	\begin{equation}
		\label{eq:lc}
		\g = \sum_x e_\g^x\, \f^x + c_\g\,\u\ .
	\end{equation}
	That is, any OPF $\g$ can be written as an $\mathbb R$-linear combination of $\{\f^1, \ldots, \f^k, \u\}$, as in \eqref{eq:lc}. Since every element of $\complex \F_d$ is a complex-linear combination of such OPFs, every such element must thus be a complex-linear combination of $\{\f^1, \ldots, \f^k, \u\}$.
\end{proof}

\begin{Corollary}
	The ``possibility of state estimation" assumption implies that, for all finite $d$, the linear space $\complex \F_d$ is finite-dimensional.
\end{Corollary}

In what follows, we introduce a representation of pure states $\psi$ that is linearly related to outcome probabilities.
Because of this, this new representation encodes the equivalence relation between ensembles, and hence, the structure of mixed states arising from alternative measurement postulates.

\begin{Definition}
	\label{def:Omega}
	For each pure state $\psi\in \mathrm P \complex^d$ we define the linear form $\Omega_\psi: \complex \F_d \to \complex$ as
	\begin{align}
		\Omega_\psi(\f) = \f(\psi)\ ,
	\end{align}
	with the natural $\SU(d)$ action
	\begin{align}
		\label{eq:U act Omega}
		U: \Omega_\psi \mapsto \Omega_{U\psi}
		\ .
	\end{align}
\end{Definition}

This allows to write the probability of outcome $\f \in \F_d$ on ensemble $(\psi_r, p_r)$, 
\begin{equation}
	\label{eq:prob out ensemb}
	P\big(\f\big|(\psi_r, p_r) \big)
	=
	\mbox{$\sum_r$}\, p_r P(\f |\psi_r)
	=
	\omega(\f)\ ,
\end{equation}
in terms of the mixed state 
\begin{equation}
	\label{eq:mixed state}
	\omega = \mbox{$\sum_r$}\, p_r\,
	\Omega_{\psi_r}
	\ .
\end{equation}
Hence, two different ensembles corresponding to the same mixed state \eqref{eq:mixed state} are indistinguishable.
The next lemma gives us important information about the group representation $\complex \F_d$.

\begin{Lemma}\label{lemma:decomp final}
	The $\SU(d)$ action \eqref{SUd action F} on $\complex\F_d$ decomposes as
	\begin{equation}
		\label{eq:decomp final}
		\complex \F_d
		\cong
		\bigoplus_{j \in \mathcal J}
		\D_j^d\ ,
	\end{equation}
	where $\D_j^d$ are the irreducible representations defined in Lemma \ref{lemma: M decomp D}.
	The finite set $\mathcal J$ contains zero and some positive integers (with no repetitions). 
\end{Lemma}

Before proving the above we mention that the quantum case is $\mathcal J= \{0,1\}$, and in section Non-quantum measurement postulate
violating associativity of the main text the (non-quantum) case $\mathcal J= \{0,1,2\}$ is analyzed.
Also, we have to mention that in this work we follow the notation of \cite{Fulton91}, where the group representations are labelled by the subspace they act on.

\begin{proof}
	In this proof we establish the following  four facts in the same order: (i) $\complex \F_d$ decomposes into a finite sum of finite-dimensional irreducible representations (irreps), (ii) these irreps are of the type $\D_j^d$, (iii) there are no repetitions, (iv) $j=0$ is always included.
	
	Fact (i). Lemma \ref{lemma:finite k} shows that the $\SU(d)$ representation $\complex \F_d$ is finite-dimensional. 
	And these can always be decomposed into finite-dimensional irreps \cite{Fulton91}.
	Also, we know that each finite-dimensional irrep of $\SU(d)$ corresponds to a $d$-row Young diagram $\lambda$.
	Hence we write
	\begin{equation}
		\label{eq:decomp proof 1}
		\complex \F_d 
		\cong 
		\bigoplus_\lambda
		\V_\lambda^d\ ,
	\end{equation}
	where repeated values of $\lambda$ can happen. 
	
	Fact (ii).
	The fact that different elements of $\complex  \F_d$ are different functions $\mathrm P \complex^d \to \complex$ implies that the form $\Omega_\psi$ (Definition \ref{def:Omega}) has support in each sub-space of \eqref{eq:decomp proof 1}. 
	Indeed, if $\Omega_\psi$ had no support in the sub-space $\V_\lambda^d$, then any of the elements $\f+\V_\lambda^d \subseteq \complex \F_d$ would correspond to the same function.
	
	Denote by $\SU(d, \psi)$ the subgroup of unitaries that leave the state $\psi$ invariant $U\psi =\psi$, and note that
	\begin{equation}
		\SU(d, \psi) 
		\cong 
		\mathrm U(1) \times \SU(d-1)\ ,
	\end{equation}
	for any $\psi$.
	According to \eqref{eq:U act Omega}, the action of $\SU(d, \psi)$ on a subspace $\V_\lambda^d$ of \eqref{eq:decomp proof 1} leaves the projection of $\Omega_\psi$ onto the subspace $\V_\lambda^d$ invariant.
	This implies that $\V_\lambda^d$ contains an invariant vector under the action of $\SU(d, \psi)$.
	But Lemma 1 from \cite{Galley_classification_2017} tells us that the only $\SU(d)$ irreps with an $\SU(d, \psi)$-invariant vector are $\D_j^d$ for $j=0,1,2,\ldots$
	Hence, all irreps $\V_\lambda^d$ in \eqref{eq:decomp proof 1} are of the form $ \D_j^d$.
	At this point it is worth mentioning that the irreps $ \D_j^d$ are real.
	
	In addition, Lemma 1 from \cite{Galley_classification_2017} tells us that in each irrep $\D_j^d$ the $\SU(d, \psi)$-invariant subspace has dimension one.
	Which fixes the projection of the linear form $\Omega_\psi$ onto each subspace of \eqref{eq:decomp final} up to a proportionality factor.
	Changing these proportionality factors modifies the structure of $\F_d$ by the corresponding inverse linear transformation; but the space $\complex \F_d$ remains identical.
	
	To prove Fact (iii), suppose that there are two repeated irreps in \eqref{eq:decomp final}.
	We can write the isomorphism 
	\begin{equation}
		\label{eq:repeated irreps}
		\D_j^d \oplus \D_j^d 
		\cong 
		\D_j^d \otimes \complex^2\ ,
	\end{equation}
	with the understanding that the $\SU(d)$-action in $\complex^2$ is trivial.
	Next, we invoke the above-shown unicity of $\Omega_\psi$ to see that the projection of $\Omega_\psi$ onto the subspace $\D_j^d \otimes \complex^2$ is of the form
	\begin{equation}
		\Omega_\psi |_{\D_j^d \otimes \complex^2}
		= 
		\Omega_\psi |_{\D_j^d} 
		\otimes 
		\Gamma |_{\complex^2}\ ,
	\end{equation}
	where $\Gamma: \complex^2 \to \complex$ is a linear form that depends on the above-mentioned proportionality factors.
	Given $\Gamma$ it possible to find two different vectors $\mathbf v, \mathbf v' \in \complex^2$ such that $\Gamma (\mathbf v) = \Gamma(\mathbf v')$.
	Then, taking any $\f\in \D_j^d$ we can construct two different elements of $\complex \F_d$ corresponding to the same function
	\begin{equation}
		(\f\otimes \mathbf v)(\psi)
		=
		(\f\otimes \mathbf v')(\psi)
	\end{equation}
	for all $\psi$, which is a non-sense.
	
	To establish Fact (iv), we recall that the unit function $\u \in \F_d$ is always included (Definition \ref{def:F_d}).
	Since $\u$ is invariant under the action \eqref{SUd action F} the trivial irrep $\D_0^d$ must be included in the decomposition \eqref{eq:decomp final}. Hence $0\in \mathcal J$. 
\end{proof}

\subsection{The $\SU(d)$ representations $\M_n^d$ and $\D_n^d$}
\label{sec:irreps M N}

\noindent
In this subsection we introduce two families of $\SU(d)$ representations that allow to construct all alternative measurement postulates for single systems by using \eqref{eq:decomp final}.
For this, we recall that the projector onto the symmetric subspace of $(\complex^d)^{\otimes n}$ can be written as the average of all permutations $\pi$ over $n$ objects
\begin{equation}
	\label{symm proj}
	P_+ = \frac 1 {n!} \sum_{\pi} \pi \ ,
\end{equation}
where $\pi$ acts by permuting the $n$ factor spaces of $(\complex^d)^{\otimes n}$.
Next we define an $\SU (d)$ representation that sometimes is named $\mathrm{Sym}^n \complex^d \otimes \mathrm{Sym}^n \complex^{d*}$.

\begin{Definition}
	Let $\M_n^d$ be the linear space of complex matrices $M$ acting on $(\complex^d)^{\otimes n}$ whose support is contained in the symmetric subspace 
	\begin{equation}
		P_+ M=MP_+ =M\ .
	\end{equation}
	And let the linear action of $\SU (d)$ on $\M_n^d$ be
	\begin{equation}
		\label{action M}
		M \mapsto U^{\otimes n}M U^{\otimes n \dagger}
		\ .
	\end{equation}
\end{Definition}

\begin{Lemma}\label{lemma: M decomp D}
	The decomposition of $\M_n^d$ into $\SU (d)$ irreducible representations is
	\begin{align}
		\label{eq:Dsubs}
		\mathcal M_n^d  =
		\bigoplus_{j=0}^n \D_{j,n}^d
		\ ,
	\end{align}  
	where the subspace $\D_{j,n}^d$ is generated by applying the group action \eqref{action M} to the element
	\begin{equation}
		\label{Def N generator}
		N_{j,n} = 
		P_+ \left( |0\rangle\! \langle 1|^{\otimes j} \otimes \unity^{\otimes (n-j)} \right) P_+
		\in \M_n^d
		\ ,
	\end{equation}
	where $|0\rangle, |1\rangle \in \complex^d$ are any orthogonal pair.
	Also, the representation isomorphisms 
	\begin{equation}
		\label{D isomorphism}
		\D^d_{j,n} \cong \D^d_{j,n'}
	\end{equation}
	hold for all $n,n' \geq j$.
\end{Lemma}

Isomorphism \eqref{D isomorphism} allows us to use the shorthand notation $\D^d_j$. Also, note that $\D^d_0$ is the trivial irrep, generated by the element $N_{0,n} = P_+ \in \M_n^d$; and $\D^d_1$ is the adjoint (quantum) irrep.


\begin{proof}
	In order to obtain the decomposition \eqref{eq:Dsubs} it is useful to define the trace map
	\begin{align}
		\label{contraction map I}
		\tr_n : \M_n^d &\to \M_{n-1}^d\ ,
		\\ 
		M &\mapsto \tr_n M\ ,
	\end{align}
	where $\tr_n$ denotes the trace over the $n$th factor in $(\complex^d)^{\otimes n}$. 
	Note that, by symmetry, this partial trace is independent of the choice of factor: $\tr_n M = \tr_1 M$.
	From now on, wherever is clear, we leave the dependence on $d$ implicit.
	
	Because the map \eqref{contraction map I} commutes with the $\SU (d)$ action,
	\begin{equation}
		\label{I commutativity}
		\tr_n\! \left[ U^{\otimes n} MU^{\otimes n \dagger} \right] 
		= U^{\otimes (n-1)} \tr_n [M]\, U^{\otimes (n-1) \dagger}
		\ ,
	\end{equation}  
	Schur's Lemma tells us that its kernel must be a subrepresentation of $\M_n$, which we denote by $\D_{n,n}$.
	It is proven in Lemma \ref{lemma:kernel irrep} that this representation is irreducible. 
	Also, it is straightforward to check that the matrix $N_{n,n}$ defined in \eqref{Def N generator} is in the kernel of the map \eqref{contraction map I}, that is 
	\begin{equation}
		\label{eq: Nnn in kernel}
		\tr_n N_{n,n}=0\ .
	\end{equation}
	Combining the above with irreducibility we see that the subspace $\D_{n,n}$ is generated by the action of the group on the single element $N_{n,n}$. 
	
	Because the map \eqref{contraction map I} is surjective, the orthogonal complement of $\D_{n,n} \subseteq \M_n$ is a representation isomorphic to $\M_{n-1}$, which in turn contains the irreducible representation $\D_{n-1,n-1} \subseteq \M_{n-1}$ in the kernel of the trace map $\tr_{n-1}: \M_{n-1} \to \M_{n-2}$. 
	Then, using Schur's Lemma again, there must be a subrepresentation $\D_{n-1,n} \subseteq \M_n$ that is isomorphic to $\D_{n-1,n-1} \subseteq \M_{n-1}$, which proves isomorphism \eqref{D isomorphism}.
	Proceeding inductively, we obtain the full decomposition \eqref{eq:Dsubs}.
	
	To conclude the proof of Lemma \ref{lemma: M decomp D} we need to show that $N_{j,n} \in \D_{j,n}$. 
	By noting that
	\begin{align}
		\tr_{n} [N_{j,n}] \propto N_{j,n-1} 
		\in \D_{j,n-1}
	\end{align}
	is non-zero when $j<n$, we can proceed inductively to arrive at
	\begin{align}
		\left(\tr_{j+1} \cdots \tr_{n-1} \tr_{n}\right) 
		\![N_{j,n}] 
		\propto N_{j,j} 
		\in \D_{j,j}\ ,
	\end{align}
	which is the case analyzed above \eqref{eq: Nnn in kernel}.
	The isomorphisms \eqref{D isomorphism} provided by Schur's Lemma conclude the proof.
\end{proof}

\subsection{The form $\Omega_\psi$ in $\M_n^d$ and $\D_n^d$}

\noindent
In this section we introduce a simple choice for the linear form $\Omega_\psi$ of Definition \ref{def:Omega}, for the cases $\complex\F_d \cong \M_n^d$ and $\complex\F_d \cong \D_n^d$.
As already mentioned, this form encodes the structure of the set of mixed states.


\begin{Lemma}
	The linear form $\Omega_\psi : \M_n^d \to \complex$ defined by 
	\begin{equation}
		\label{eq:Ansatz}
		\Omega_\psi (M)
		=
		\tr\!\left( |\psi\rangle\! \langle\psi |^{\otimes n} M \right)\ ,
	\end{equation}
	is invariant under all stabilizer unitaries $U\in \SU(d,\psi)$
	\begin{equation}
		\label{eq:stab invariance}
		\Omega_\psi (U^{\otimes n} M U^{\otimes n \dagger})
		=
		\Omega_\psi (M)\ ,
	\end{equation}
	and has support in all irreps $\D^d_{j} \subseteq \M^d_n$. 
\end{Lemma}

\begin{proof}
	It is straightforward to check that the form \eqref{eq:Ansatz} satisfies \eqref{eq:stab invariance}.
	To see that \eqref{eq:Ansatz} has support in each irrep $\D_{j,n} \subseteq \M_n$, we observe that, for each $j$, there is a pure state $\psi$ such that
	\begin{equation}
		\label{full support}
		\tr\!\left[ N_{j,n} 
		|\psi\rangle\! \langle\psi |^{\otimes n} \right]
		\neq 0\ ,
	\end{equation}
	where $N_{j,n}$ is defined in \eqref{Def N generator}.
\end{proof}

As mentioned above, these two constrains fix $\Omega_\psi$ up to an irrelevant proportionality factor in each irrep.
To obtain $\Omega_\psi$ in the case $\complex \F_d = \D_n^d$ we proceed in the following maner.
Since $\D_n^d$ is a subrepresentation of $\M_n^d$ we can take \eqref{eq:Ansatz} and perform the orthogonal projection onto the subspace $\D_{n,n}^d \subseteq \M_n^d$, defined via   \eqref{Def N generator} or via the kernel fo the map \eqref{contraction map I}.

\section{Multipartite systems}
\label{app:OPFcomp}

\noindent
In this section we describe and impose the consistency constraints associated to composite systems and the star product.

\subsection{Closedness under system composition}

\noindent
We require that any family of OPF sets $\F_2, \F_3, \ldots$ and $\F_\infty$ must be closed under system composition.
This means that the complete set of measurements $\F_a$ of a system $\complex^a$ also includes the measurements that appear in the description of $\complex^a$ as part of a larger system $\complex^a \otimes \complex^b$.

\begin{Definition}[Closedness under system composition]\label{def:closedness}
	If $\F_{ab}$ is the OPF set of $\complex^a \otimes \complex^b$ then the OPF set $\F_a$ of $\complex^a$ is the following collection of functions
	\begin{align}  
		\label{eq:Fa}
		\mathrm P \complex^a &\to [0,1]
		\\ \label{eq:Fab to Fa}
		\alpha &\mapsto \f(\alpha \otimes \beta)
	\end{align}
	for all $\f \in \F_{ab}$ and a fixed $\beta \in \mathrm P \complex^b$, and all $a,b \in \{2,3, \ldots, \infty\}$.
\end{Definition}

Note that the closure of $\F_{ab}$ under $\unity \otimes \SU(b)$ implies that the set $\F_a$ defined via (\ref{eq:Fa}-\ref{eq:Fab to Fa}) does not depend on the choice of $\beta$.
Also, it is straightforward to check that the OPF set $\F_a$ so defined satisfies all the requirements of Definition~\ref{def:F_d}.

\subsection{$\complex \F_d \cong \M_n^d$}

\noindent
In the finite-dimensional case, closedness under system composition (Definition~\ref{def:closedness}) implies the following strong fact.
For any set of measurements $\F_{ab}$ of a bipartite system $\complex^a \otimes \complex^b$, the measurement spaces of the subsystems are $\complex \F_a \cong \M^a_n$ for $\complex^a$ and $\complex \F_b \cong \M^b_n$ for  $\complex^b$, with the same $n$.
In addition, using the fact that any pair of systems can be jointly described as a bipartite system, we conclude that all finite-dimensional systems $\complex^d$ must have OPF space $\M^d_n$ (with the same value for $n$). 

\begin{Lemma}\label{F=Mn}
	For any pair of positive integers $a,b$, let $\F_{ab}$ be the OPF set of $\complex^a \otimes \complex^b$ with decomposition (Lemma~\ref{lemma:decomp final})
	\begin{equation}
		\complex\F_{ab}
		\cong
		\bigoplus_{j\in \mathcal J}
		\D_j^{ab}\ .
	\end{equation}
	Define $\F_a$ as the set of functions
	\begin{align}  
		\mathrm P \complex^a &\to [0,1]
		\\ \label{Fab to Fa}
		\alpha &\mapsto \f(\alpha \otimes \beta)
	\end{align}
	for all $\f \in \F_{ab}$ and a fixed $\beta \in \mathrm P \complex^b$.
	Then we have the $\SU(a)$-representation isomorphism
	\begin{equation}
		\label{iso:desired}
		\complex \F_a \cong \M^a_{n}
		\ ,
	\end{equation}
	where $n= \max \mathcal J$.
\end{Lemma}

Note that, if we define $\F_b$ by exchanging the role of the subsystems $\complex^a \otimes \complex^b$ in \eqref{Fab to Fa}, then we obtain the $\SU(b)$-representation isomorphism $\complex \F_b \cong \M^b_{n}$ with the same value for $n$ as in \eqref{iso:desired}.
Using the fact that any pair of systems can be jointly described as a bipartite system, we arrive at the following.

\begin{Corollary}
	Closedness under system composition (Definition~\ref{def:closedness}) implies that all finite-dimensional Hilbert spaces $\complex^d$ have OPF space $\complex \F_d \cong \M_n^d$ with the same $n$.
\end{Corollary}

\begin{proof}[Proof of Lemma \ref{F=Mn}]
	In order to establish the isomorphism \eqref{iso:desired} we analyze how the functions \eqref{Fab to Fa} transform under the subgroup $\SU(a) \otimes \unity$.
	First, we do this in the case where $\mathcal J$ has finite cardinality, so that the decomposition \eqref{eq:decomp final} of $\complex\F_{ab}$ has a largest irrep $\D_n^{ab}$.
	We further split this analysis into the case where the function $\f$ in \eqref{Fab to Fa} belongs to the subspace $\f \in \D_n^{ab} \subseteq \complex \F_{ab}$, and the general case.
	
	Using the characterization of $\D_n^{ab}$ as the kernel of the map \eqref {contraction map I} we can say the following.
	For each $\f \in \D_n^{ab} \subseteq \complex \F_{ab}$ there is a matrix $F\in \M^{ab}_n$ such that $\tr_n F =0$ and
	\begin{equation}
		\label{eq:N subspace}
		\f(\alpha \otimes \beta)
		=
		\tr\!\left(F \,
		|\alpha\rangle\! \langle \alpha|^{\otimes n}\!
		\otimes\!
		|\beta\rangle\! \langle \beta|  ^{\otimes n} \right)\ ,
	\end{equation}
	for all $\alpha, \beta$.
	(Note that, in order to improve clarity, we re-arranged the order of the tensor factors.)
	The matrices
	\begin{align}
		\label{eq: alpha beta}
		|\alpha\rangle\! \langle \alpha|^{\otimes n}\!
		\otimes\!
		|\beta\rangle\! \langle \beta|^{\otimes n}  
		\in
		\M_n^{ab}\ ,
	\end{align}
	are contained in the subspace
	\begin{align}
		\nonumber
		|\alpha\rangle\! \langle \alpha|^{\otimes n}\!
		\otimes\!
		|\beta\rangle\! \langle \beta|^{\otimes n}  
		&\in
		P^\A_+ P_+^\B \M_n^{ab} P^\A_+ P_+^\B
		\\ \nonumber
		& \cong 
		\M_n^a \otimes \M_n^b
		\\ & \cong
		\bigoplus_{j,j'=0}^n \D_{j}^a 
		\otimes \D_{j'}^b
		\ ,
		\label{eq: prod space}
	\end{align}
	where the isomorphisms are of $\SU(a)\otimes \SU(b)$ representations.
	Even more, using the full support conditions \eqref{full support} in each tensor factor, we conclude that the matrices $|\alpha\rangle\! \langle \alpha|^{\otimes n}\!
	\otimes\!
	|\beta\rangle\! \langle \beta|^{\otimes n}$ generate the whole space \eqref{eq: prod space}.
	
	Next we analyze the $\SU(a) \otimes \SU(b)$ action on the function \eqref{eq:N subspace}, which is the action on the intersection between the subspaces $\{ F\in \M^{ab}_n : \tr_n F =0 \}$ and \eqref{eq: prod space}.
	This intersection can be characterized by writing the trace as 
	\begin{equation}
		\tr_n = \tr_{\A_n} \tr_{\B_n}\ ,
	\end{equation}
	where $\tr_{\A_n}$ is the trace on the $n$th factor of $\M_n^a$, and $\tr_{\B_n}$ is the trace on the $n$th factor of $\M_n^b$.
	The above identity implies that if $\tr_{\A_n} F=0$ or $\tr_{\B_n} F=0$ then $\tr_n F=0$. Therefore, the above-mentioned intersection contains all irreps $\D^a_n \otimes \D^b_j$ and $\D^a_j \otimes \D^b_n$ for $j=0,1,\ldots, n$.
	This implies that the $\SU(a) \otimes \unity$ action on the space of functions \eqref{eq:N subspace} with $\f \in \D_n^{ab} \subseteq \complex \F_{ab}$ decomposes into the irreps $\D_0^a, \ldots, \D_n^a$, with possible repetitions.
	
	In the general case $\f \in \complex \F_{ab}$, the addition of all subspaces $\D_j^{ab} \subseteq \F_{ab}$ with $j<n$ does not add any new irrep to the list $\D_0^a, \ldots, \D_n^a$.
	Although it may increase the repetitions.
	
	Finally, we establish the desired isomorphism \eqref{iso:desired} by recalling Lemma \ref{lemma:decomp final}. This tells us that any OPF set, like the $\F_a$ defined through \eqref{Fab to Fa}, has no repeated irreps.
\end{proof}

\subsection{The star product}

\noindent
In this subsection we introduce the star product, which contains the information of which measurements of a composite system $\F_{ab}$ are local.

\begin{Definition}\label{def:star-product}
	The star product is a map $\star: \F_a \times \F_b \to \F_{ab}$ defined on any pair of OPF sets $\F_a, \F_b$, with the following properties:
	\begin{itemize}
		
		\item preserves the local structure
		\begin{equation}
			\label{eq:factorization constraint 0}
			(\f\star \u)(\alpha \otimes \beta) = 
			\f(\alpha)\ ,
		\end{equation}
		
		\item preserves probability
		\begin{align}
			\label{eq:u*u=u}
			\u_\A \star \u_\B &= \u_{\A\B}\ ,
			\\ \label{cond:0=0}
			\f_\A \star \mathbf 0_\B &= \mathbf 0_{\A\B}
		\end{align}
		
		\item commutes with local mixing operations
		\begin{equation}
			\label{comm mix proc}
			\left(\mbox{$\sum_x$}\, p_x\, \f^x\right) \star \g
			=
			\mbox{$\sum_x$}\, p_x \left(\f^x \star \g \right)\ ,
		\end{equation}
		
		\item commutes with the local group action
		\begin{align}
			\label{eq:commutativity}
			(\f \circ U)\star \g =
			(\f\star \g)\circ (U\otimes \unity)\ ,
		\end{align}
		
		\item and it is associative
		\begin{equation}
			\label{associativity}
			(\f\star\g)\star \mathbf h
			=
			\f\star(\g\star \mathbf h)\ ,
		\end{equation}
	\end{itemize}
	for any $\f,\f^x\in \F_a$; $\g\in \F_b$; $\mathbf h \in\F_c$; $\alpha\in \mathrm P \complex^a$; $\beta \in \mathrm P \complex^b$; $U\in \U(a)$; $a,b \in \{2,3,\dots, \infty\}$, and any probability distribution $p_x$.
	Properties (\ref{eq:factorization constraint 0}-\ref{associativity}) must also hold when exchanging factors.
\end{Definition}

The $\star$-product allows us to write the reduced state of a bipartite pure state $\psi \in \complex^a \otimes \complex^b$ on the subsystem $\complex^a$ as the linear form $\f \mapsto \Omega_{\psi} (\f\star \u)$ for all $\f\in \complex\F_a$.

Note that property \eqref{eq:factorization constraint 0} is weaker than the analog condition in the main text:
\begin{equation}
	\label{eq:factorization stronger}
	(\f\star \g)(\alpha \otimes \beta) = 
	\f(\alpha) \g(\beta)\ .
\end{equation}
The reason for writing the stronger condition in the main text is that it does not require $\u$ to be defined.
The following lemma proves that, in our context, condition \eqref{eq:factorization constraint 0} implies condition \eqref{eq:factorization stronger}.

\begin{Lemma}
	Suppose that any ensemble $(\psi_r, p_r)$ satisfying
	\begin{equation}
		\sum_r p_r\, \f(\psi_r)
		=
		\f(\varphi)\ ,
		\ \ \ \forall\f\in\F_d\ ,
	\end{equation}
	is of the form $\psi_r= \varphi$ for all $r$.
	Then \eqref{eq:factorization constraint 0} implies
	\begin{equation}
		\label{eq:factorization constraint}
		(\f\star \g)(\alpha \otimes \beta) = 
		\f(\alpha)\, \g(\beta)\ .
	\end{equation}
\end{Lemma}

\begin{proof}
	First, let $\{\g_i \}$ be a complete measurement and define the following probabilities and the (not necessarily pure) states
	\begin{align}
		p_i &= \g_i(\beta) \ ,
		\\
		\Omega_i (\f) &= \left\{ \begin{array}{ll}
			(\f \star \g_i)(\alpha \otimes \beta) 
			/p_i
			& \mbox{ if } p_i \neq 0
			\\
			\f(\alpha) & \mbox{ if } p_i=0
		\end{array}\right. \ ,
	\end{align}
	for all $\f$.
	Second, substitute $\sum_i \g_i =\u$ in \eqref{eq:factorization constraint 0} obtaining
	\begin{equation}
		\sum_i p_i\, \Omega_i(\f) 
		=
		\f(\alpha)\ ,
	\end{equation}
	for all $\f$.
	Third, use the premise of the lemma to conclude that
	\begin{equation}
		\Omega_i(\f) 
		=
		\f(\alpha)
	\end{equation}
	for all $i$ and $\f$.
	Finally, substituting back the definition of $\Omega_i$ we obtain
	\begin{equation}
		(\f \star \g_i)(\alpha \otimes \beta)
		= \f(\alpha)\, \g_i(\beta)\ ,  
	\end{equation}
	which implies \eqref{eq:factorization constraint}.
\end{proof}

By ``preservation of probability" \eqref{eq:u*u=u} it is meant that the fact that all outcome probabilities add up to one
\begin{align}
	\sum_i \f_i (\alpha) &= 1\ ,
\end{align}
is independent of whether we describe a system on its own or as part of a larger system
\begin{equation}
	\sum_{i} (\f_i \star \u_\B)
	(\psi_{\A\B}) = 1\ .
\end{equation}
Also, the joint outcome $\f_\A \star \bf 0_\B$, where $\bf 0_\B$ is the formal outcome with zero probability for all states, must have zero probability, which gives \eqref{cond:0=0}.

The action of the $*$-product is not defined on the elements of $\complex \F_a$ that are not in $\F_a$.
However, the following lemma shows that one can define the action of the $*$-product to the rest of elements of $\complex \F_a$ in such a way that the map is bilinear.

\begin{Lemma}
	Any star-product map $\star: \F_a \times \F_b \to \F_{ab}$ as specified in Definition \ref{def:star-product} can be extended to a bilinear map $\star: \complex \F_a \times \complex \F_b \to \complex \F_{ab}$ with the same properties (\ref{eq:factorization constraint 0}-\ref{associativity}).
\end{Lemma}

\begin{proof}
	For any given $\g\in \F_b$ define the map
	\begin{align}
		\xi : \F_a &\rightarrow \F_{ab} \ ,\\
		\f &\mapsto \f \star \g \ .
	\end{align}
	Using Definition \ref{def:star-product} we obtain the following properties for the map
	\begin{align}
		&\xi (\mathbf 0) = \mathbf 0\ ,
		\\
		&\xi (\mbox{$\sum_x$}\, p_x\, \f^x )
		=
		\mbox{$\sum_x$}\, p_x\, \xi (\f^x)\ ,
	\end{align}
	for any probability distribution $p_x$.
	In Appendix~1 of \cite{Hardy_quantum_2001} it is proven that it is possible to define a $\mathbb R$-linear map $\xi': \mathbb R\F_a \to \mathbb R\F_{ab}$ which is identical to $\xi$ inside $\F_a$.
	Finally, we can define the $\complex$-linear map $\xi'': \complex \F_a \to \complex \F_{ab}$ in the natural way
	\begin{equation}
		\xi'' (\f_1 +\mathrm i \f_2)
		=
		\xi' (\f_1) +\mathrm i\, \xi'(\f_2)
	\end{equation}
	for any pair $\f_1, \f_2 \in \mathbb R \F_a$.
	
	The above construction can be repeated with an exchange of parties. Proving the desired result.
\end{proof}

\begin{Lemma}\label{Lemma:product measurements}
	In the case $\complex \F_d \cong \M^d_{n}$ we have the identity
	\begin{equation}
		\label{eq:++ subspace}
		P^\A_+ P_+^\B 
		\left(\M_n^a \star \M_n^b\right)
		P^\A_+ P_+^\B
		= 
		\M_n^a \otimes \M_n^b\ ,
	\end{equation}
	of $\SU(a) \otimes \SU(b)$ representations.
\end{Lemma}

\begin{proof}
	If we write condition \eqref {eq:factorization constraint} with the form $\Omega_\alpha$ introduced in \eqref{eq:Ansatz} then we get
	\begin{align}
		\nonumber
		&\tr\!\left(
		[F \star G] |\alpha\rangle\! \langle \alpha|^{\otimes n}\!
		\otimes\!
		|\beta\rangle\! \langle \beta|^{\otimes n}
		\right)
		\\ \nonumber &\qquad\  
		=
		\tr\!\left(
		F |\alpha\rangle\! \langle \alpha|^{\otimes n}
		\right)
		\tr\!\left(
		G |\beta\rangle\! \langle \beta|^{\otimes n}
		\right)
		\\ \label{eq:59} &\qquad\  
		=
		\tr\!\left(
		[F \otimes G] |\alpha\rangle\! \langle \alpha|^{\otimes n}\!
		\otimes\!
		|\beta\rangle\! \langle \beta|^{\otimes n}
		\right)\ ,
	\end{align}
	for all $F\in \M_n^a$, $G\in \M_n^b$, $\alpha\in \mathrm P \complex^a$ and $\beta \in \mathrm P \complex^b$.
	By noting that the set of matrices $|\alpha\rangle\! \langle \alpha|^{\otimes n}\! \otimes\! |\beta\rangle\! \langle \beta|^{\otimes n}$ span the subspace $P^\A_+ P_+^\B \M_n^{ab} P^\A_+ P_+^\B \subseteq \M_n^{ab}$ we can write \eqref{eq:59} as
	\begin{equation}
		P^\A_+ P_+^\B 
		\left(F \star G\right)
		P^\A_+ P_+^\B
		=
		F \otimes G\ ,
	\end{equation}
	for all $F,G$.
	This proves identity \eqref{eq:++ subspace}.
\end{proof}

\section{Non-associativity of $\star:\M_a \times \M_b \to \M_{ab}$}
\label{app:nonassoc}

\subsection{The permutation group and Schur-Weyl duality}

\noindent
In this section we review some well-known results of representation theory. %
The $n$-th tensor-power of a vector space  $\complex^d$ can be decomposed as
\begin{align}
	\left( \complex^d \right)^{\otimes n}
	\cong \bigoplus_\lambda 
	\V_\lambda^d \otimes \mathcal S_\lambda^n
	\ ,
\end{align}
where $\lambda$ runs over all partitions of $n$ with at most $d$ parts, $\V^d_\lambda$ are irreps of $\SU(d)$, and $\mathcal S^n_\lambda$ are the irreps of the group of permutations of $n$ objects.
The partition $\lambda=(n)$ corresponds to the trivial representation of the group of permutations, and hence, all vectors in the subspace $\V_{(n)}^d \otimes \mathcal S_{(n)}^n$ are permutation-invariant. 
Because of this, this subspace and the corresponding projector $P_+=P_{(n)}$ are called symmetric.

When considering a bipartite space $\complex^d = \complex^a \otimes \complex^b$, the symmetric projector can be written as
\begin{align}
	\label{AB symm proj}
	P_+^{\A\B}
	= \sum_\lambda 
	Q_\lambda^{\A\B}\ ,
\end{align}
where $Q_\lambda^{\A\B}$ is the orthogonal projector onto the subspace of $[\V_\lambda^a \otimes \mathcal S_\lambda^n]^\A \otimes [\V_\lambda^b \otimes \mathcal S_\lambda^n]^\B$ that transforms trivially when applying the same permutation to $\A$ and $\B$.
Specifically, we can write it as
\begin{equation}
	Q^{\A\B}_\lambda 
	=
	\unity_{\V^a_\lambda} \otimes
	\unity_{\V^b_\lambda} \otimes
	|\tau\rangle
	_{\mathcal S_\lambda^n\mathcal S_\lambda^n}
	\langle \tau|
\end{equation}
where $\unity_\V$ is the identity on the subspace $\V$ and 
\begin{equation}
	|\tau_\lambda\rangle _{\V\V'} 
	= 
	\sum_{k}
	|k\rangle_{\V} \otimes |k\rangle_{\V'}
\end{equation}
is the ``maximally entangled state" of the product space $\V\otimes \mathcal \V'$.
The invariance of $|\tau_\lambda\rangle _{\A\B}$ under identical permutations on $\A$ and $\B$ is analogous to the invariance of any maximally entangled state under transformations of the form $U\otimes U^*$, together with the fact that all irreps of the permutation group are real (self-dual).

In the tri-partite case $\complex^d = \complex^a \otimes \complex^b \otimes \complex^c$, the symmetric projector can be written as
\begin{equation}
	\label{eq:ABC + proj}
	P_+^{\A\B\C} 
	= 
	\sum_{\lambda, \mu, \nu}  
	Q_{\lambda, \mu, \nu} ^{\A\B\C} \ ,
\end{equation}
where $Q_{\lambda, \mu, \nu} ^{\A\B\C}$ is the orthogonal projector onto the subspace of $[\V_\lambda^a \otimes \mathcal S_\lambda^n]^\A \otimes [\V_\mu^b \otimes \mathcal S_\mu^n]^\B \otimes [\V_\nu^c \otimes \mathcal S_\nu^n]^\C$ that transforms trivially when applying the same permutation on $\A,\B,\C$.
Therefore, the projector $Q_{\lambda, \mu, \nu} ^{\A\B\C}$ is zero unless the irrep decomposition of $\mathcal S_\lambda^n \otimes \mathcal S_\mu^n \otimes \mathcal S_\nu^n$ contains the trivial $\mathcal S^n_{(n)}$.
Particularly, when one of the three partitions is $(n)$, we recover the bipartite case 
\begin{equation}
	\label{eq:Q n lambda lambda}
	Q^{\A\B\C} _{(n), \lambda, \lambda } 
	=   
	\unity_{\V^a_{(n)}} \otimes 
	\unity _{\V^b_\lambda} \otimes 
	\unity _{\V^c_\lambda} \otimes 
	|\tau\rangle ^{\B \C} 
	_{\mathcal S_\lambda^n\mathcal S_\lambda^n}\!
	\langle \tau|
	\ ,
\end{equation} 
for all $\lambda$.
That is, if one partition is $(n)$ then the other two partitions have to be equal. And this is why, in the bipartite case, the projector $Q^{\A\B}_\lambda$ only depends on one partition.

\subsection{The irreducible representations of $\SU(d)$}

\noindent
Using the Littlewood-Richardson rule \cite{Fulton91}, we can decompose $\V^d_\lambda \otimes \V^{d}_{\mu}$ into irreps, and prove the following patterns.

\begin{Lemma}\label{decomp L L*}
	The decomposition of $\V^4_{(n-1,1)} \otimes \V^{4*}_{(n-1,1)}$ into irreps does not include any $\D^4_j$ with $j\geq n$.
\end{Lemma}

\begin{proof}
	If we denote by $\lambda^*$ the partition of the irrep $\V^{4*}_{(n-1,1)}$, and by $\lambda_j$ the partition of $\D^4_j$, then we have
	\begin{align} 
		\lambda^* &= (n-1, n-1, n-2)\ ,
		\\
		\lambda_j &= (2j,j,j)\ ,
		\\
		\lambda &= (n-1,1)\ .
	\end{align}
	Applying the Littlewood-Richardson rule \cite{Fulton91} to the Young tableaux $\lambda$ and $\lambda^*$, we see that all resulting tableaux have at most $4(n-1)$ boxes, while the tableau of $\lambda_j$ has $4j$ boxes.
	Therefore, no tableau $\lambda_j$ with $j\geq n$ can appear in the product of $\lambda$ and $\lambda^*$.
\end{proof}

\begin{Lemma}\label{su2 in su4}
	When restricting the irrep $\D_n^4$ of $\SU(4)$ to any $\SU(2)$ subgroup, the decomposition of $\D_n^4$ into irreps of $\SU(2)$ does not include any $\D^2_j$ with $j>n$.
\end{Lemma}

\begin{proof}
	The $\SU(2)$ irreps in $\D_n^4$ correspond to straight lines of weights in the weight diagram of $\D_n^4$. The longest such line contains $2n+1$ weights. Therefore, the largest $\SU(2)$ irrep is $\D_n^2$.
\end{proof}

\subsection{Proof of the main theorem}

\noindent
The following theorem shows that only in the quantum cas (that is $n=1$) there is an associative star product $\star: \M_n^a \times \M_n^b \to \M_n^{ab}$.

\begin{Theorem}
	If $n \geq 2$ then there is no bilinear map $\star: \M_n^a \times \M_n^b \to \M_n^{ab}$ satisfying the star-product Definition \ref{def:star-product}.
\end{Theorem}

\begin{proof}
	\emph{Analysis of bipartite systems.}
	Let us consider a bipartite system with Hilbert space $\complex^a \otimes \complex^b$ and dimensions $a=2$ and $b=4$.
	Using the decomposition \eqref{AB symm proj} of the projector $P_+^{\A\B}$ onto the symmetric subspace of $( \complex^a \otimes \complex^{b} )^{\otimes n}$ we can decompose $\M_n^{ab}$ into subspaces as
	\begin{align}
		\label{eq:decomp n}
		\M_n^{ab} &= \sum_{\lambda, \mu}
		Q^{\A\B}_\lambda \M_n^{ab} Q^{\A\B}_{\mu}
		\ ,
	\end{align}
	each labeled by a pair of $n$-partitions $(\lambda, \mu)$.
	Since system $\A$ is 2-dimensional all $n$-partitions $\lambda, \mu$ have at most two rows.
	The action of $\SU(ab)$ might not be well-defined in some of these subspaces, but the action of the subgroup $\SU (a)\otimes \SU (b) \subseteq \SU(ab)$ is well-defined in each $(\lambda, \mu)$ subspace from \eqref{eq:decomp n}.
	Concretely, we have the following isomorphism of $\SU(a) \otimes \SU (b)$ representations
	\begin{equation}
		\label{eq:subs l l'}
		Q^{\A\B}_\lambda \M_n^{ab} Q^{\A\B}_{\mu}
		\cong
		\V^a_\lambda \otimes \V^{a*}_{\mu} \otimes 
		\V^b_\lambda \otimes \V^{b*}_{\mu}\ ,
	\end{equation}
	which in particular gives
	\begin{align}
		\label{eq:++projection}
		Q^{\A\B}_{+} \M_n^{ab} Q^{\A\B}_{+}
		\cong 
		\M_n^a \otimes \M_n^b\ .
	\end{align}
	
	Now, let us take the subspace \eqref{eq:subs l l'} corresponding to $\lambda =\mu =(n-1,1)$, and decompose it into two orthogonal subspaces
	\begin{align}
		\label{eq: dec of LL*}
		Q^{\A\B}_{(n-1,1)} \M_n^{ab} Q^{\A\B}_{(n-1,1)}
		= \mathcal W^b_\mathrm{licit} 
		\oplus \mathcal W^b_\mathrm{illicit}\ ,
	\end{align}
	defined in the following way: (i) consider the action of the subgroup $\unity \otimes \SU (b)$ on the left-hand side of \eqref{eq: dec of LL*}, (ii) decompose this action into irreps of $\SU (b)$, (iii) let $\mathcal W^b_\mathrm{licit}$ be the direct sum of all irreps $\D_j^{b}$ for any $j$, and (iv) let $\mathcal W^b_\mathrm{illicit}$ be the direct sum of the rest of irreps.
	The super-index $b$ in these subspaces $\mathcal W^b_\mathrm{xxlicit}$ reminds us that these are $\SU(b)$ representations.
	
	Lemma \ref{decomp L L*} tells us that $\mathcal W^b_\mathrm{licit}$ does not contain any $\D^b_j$ with $j\geq n$. That is
	\begin{align}
		\label{eq: j<n}
		\mathcal W^b_\mathrm{licit} 
		\cong 
		\bigoplus_{j < n} \D_j^{b}  \ ,
	\end{align}
	where the sum over $j$ may contain some absences and repetitions.
	Combining this with Schur's Lemma and the commutativity constraint \eqref {eq:commutativity}, we see that the image of the $\star$-product $[\M_n^a \star \M_n^b] \subseteq \M_n^{ab}$ does not have support in $\mathcal W^b_\mathrm{illicit}$.
	In particular, 
	\begin{align}
		\label{eq:inc Ds}
		Q^{\A\B}_{(n-1,1)}
		\left[\u_\A \star \M_n^b \right]
		Q^{\A\B}_{(n-1,1)}
		\subseteq 
		\mathcal W^b_\mathrm{licit} \ .
	\end{align}
	In the next subsection we show that, if $\B$ is itself considered a bipartite system then the above subspace contains the irrep $\D_n^b$, which is incompatible with \eqref{eq: j<n} and \eqref{eq:inc Ds}.
	
	\medskip
	\emph{Analysis of tripartite systems.}
	Now let us describe system $\B$ as a bipartite system $\C\E$ with Hilbert space $\complex^b = \complex^c \otimes \complex^e$ and dimensions $c=e=2$.
	Combining the decompositions of the bipartite \eqref{AB symm proj} and the tripartite \eqref {eq:ABC + proj} symmetric projectors, we can write
	\begin{align}
		\label{eq:decomp ABC}
		Q^{\A\B}_{(n-1,1)}
		=
		\sum_{\mu, \nu}
		Q^{\A\C\E}_{(n-1,1),\mu,\nu} \ .
	\end{align}
	Now, if we substitute the decomposition \eqref{eq:decomp ABC} into \eqref{eq:inc Ds} and remove all terms except for the $\mu=(n)$ and $\nu=(n-1,1)$ one, the inclusion still holds
	\begin{align}
		\label{eq:DDDDD}
		Q^{\A\C\E}_{(n-1,1),(n),(n-1,1)}
		\left[\u_\A \star \M_n^{ce} \right]
		Q^{\A\C\E}_{(n-1,1),(n),(n-1,1)}
		\subseteq 
		\mathcal W^{ce}_\mathrm{licit} .
	\end{align}
	Importantly, the projector $Q^{\A\C\E}_{(n-1,1),(n),(n-1,1)}$ is non-zero according to \eqref{eq:Q n lambda lambda}.
	
	Now, if we restrict the action of $\unity \otimes \SU(ce)$ to the subgroup $\unity \otimes \SU(c) \otimes \unity$ and use \eqref{eq: j<n} and Lemma \ref{su2 in su4}, then we see that the irrep decompositions of the right-hand sides of \eqref{eq:DDDDD} does not contain $\D_n^c$.
	Also, due to the fact that the subspace
	\begin{equation}
		\label{eq:wrong subspace}
		Q^{\A\C\E}_{(n-1,1),(n),(n-1,1)}
		\left[\u_\A \star \M_n^c \star \u_\E\right]
		Q^{\A\C\E}_{(n-1,1),(n),(n-1,1)}
	\end{equation}
	is a subrepresentation of \eqref{eq:DDDDD}, it does not contain the irrep $\D_n^c$.
	Next we show that this is incompatible with associativity.
	Using Lemma \ref{Lemma:product measurements} and recalling that 
	\begin{equation}
		\u_{\C\E} = P_+^{\C\E} 
		= 
		\sum_\lambda Q_\lambda ^{\C\E}
	\end{equation}
	we obtain the isomorphism  \begin{align*}
		Q^{\A\C\E}_{(n),(n-1,1),(n-1,1)}
		\left[ \M_n^a \star \u_{\C\E} \right] 
		Q^{\A\C\E}_{(n),(n-1,1),(n-1,1)}
		\cong 
		\M_n^a
	\end{align*}
	of $\SU(a) \otimes \unity \otimes \unity$ representations, which include the irrep $\D_n^a$.
	By permuting the subsystems $\A\C\E$ we conclude that the $\unity \otimes \SU(c) \otimes \unity$ representation \eqref{eq:wrong subspace} also contains the irrep $\D_n^c$, in contradiction with our previous conclusion!
\end{proof}

At this point we can contrast the above argument with the disregarded case $n=1$. In this cse there is only one partition $\lambda =\mu =(1)$, and
\begin{equation*}
	Q^{\A\B}_\lambda \M_1^{ab} Q^{\A\B}_{\mu}
	= \M_1^{ab}
	= \M_1^a \otimes \M_1^b
	= \mathcal W^b_\mathrm{licit} \ ,
\end{equation*}
which implies that $\mathcal W^b_\mathrm{illicit}$ is trivial.
Therefore the above contradiction does not apply to the $n=1$ case.

\begin{Corollary}[measurement theorem]
	\label{Coro:main}
	Any family of OPF sets $\F_d$ with finite $d$, equipped with a $\star$-product, and satisfying the assumptions ``possibility of state estimation" and ``closedness under system composition", has OPFs and $\star$-product of the form
	\begin{align}
		&\f(\varphi) = 
		\langle\varphi| F|\varphi\rangle\ ,   
		\\
		&(\f\star\g)(\psi) =
		\langle\psi| F\otimes G |\psi\rangle\ ,
	\end{align}
	for all normalized $\varphi\in \complex^a$ and $\psi \in \complex^a \otimes \complex^b$, where the $\complex^a$-matrix $F$ satisfies $0\leq F \leq \unity$, and analogously for $G$.
\end{Corollary}

\section{Countably infinite-dimensional Hilbert spaces}
\label{app:inftdim spaces}

\noindent
Since all countably infinite-dimensional Hilbert spaces are isomorphic, we denote them all by $\complex^\infty$.
The topological space of all one-dimensional subspaces of $\complex^\infty$ is denoted by $\mathrm P \complex^\infty$. 
Also, for any given subspace $S\subseteq \complex^\infty$ we denote the corresponding orthogonal projector by $\Pi_S$.

The following lemma tells us that the measurements on $\complex^\infty$ are of the quantum form \eqref{eqFormOfOPFs} if and only if they have such form when restricted to any finite-dimensional subspaces of $\complex^\infty$.

	\begin{Lemma}
		\label{LemSubspace}
		For every (not necessarily continuous) function $\f : \mathrm P \complex^\infty \to [0,1]$, the following two statements are equivalent: 
		\begin{itemize}
			\item There exists a self-adjoint operator $F$ such that $0\leq F \leq \unity$ and $\f(\psi) = \langle \psi|F|\psi\rangle$ for all normalized $\psi\in \complex^\infty$.
			\item For every finite-dimensional subspace $S\subset \complex^\infty$, there exists a self-adjoint operator $F_S$ fully supported on $S$, i.e.\ $\Pi_S F_S \Pi_S = F_S$, such that $0\leq F_S \leq \unity$ and $\f(\psi) = \langle\psi| F_S |\psi\rangle$ for all normalized $\psi\in S$. 
			
		\end{itemize}
		
	\end{Lemma}
	
	\begin{proof}
		Suppose the first statement, $\f(\psi) = \langle\psi| F |\psi\rangle$. Then, for every finite-dimensional subspace $S$, define $F_S = \Pi_S F \Pi_S$. 
		Now, it is clear that for all normalized $\psi\in S$, we have 
		\begin{equation}
			\langle\psi| F_S |\psi\rangle
			= 
			\langle\psi| F |\psi\rangle
			= \f(\psi)
		\end{equation}
		which is the second statement of the lemma.

		Conversely, suppose that for every finite-dimensional subspace $S \subset \complex^\infty$ there exists $F_S$ satisfying $\Pi_S F_S \Pi_S = F_S$ and $\f(\psi) = \langle\psi| F_S |\psi\rangle$ for all normalized $\psi\in S$. First we prove the following intermediate claim: \emph{Let $(S^{(n)})_{n\in\mathbb{N}}$ be any sequence of subspaces with $\dim S^{(n)}=n$ and $S^{(n)}\subset S^{(n+1)}$ such that for $S:=\bigcup_{n\in\mathbb{N}} S^{(n)}$ we get the norm closure $\bar S=\complex^\infty$. Then there exists a unique bounded operator $F$ on $\complex^\infty$ such that $\f(\psi)=\langle\psi|F|\psi\rangle$ for all normalized states $\psi\in S$.}
		
		To prove this, note that the sequence of subspaces defines a unique orthonormal basis $\{|i\rangle\} _{i\in\mathbb{N}}$ of $\complex^\infty$ such that  $S^{(n)} = {\rm span}\{|1\rangle,|2\rangle,\ldots,|n\rangle\}$ (this follows e.g.\ from Gram-Schmidt orthogonalization). Define the projector $\Pi^{(n)} = \sum_{i=1}^n \proj i$ onto $S^{(n)}$, and define the self-adjoint operator $F^{(n)} = F_{S^{(n)}}$ whose existence we have assumed as a premise. It satisfies $\f(\psi)=\langle\psi|F^{(n)}|\psi\rangle$ for all normalized $\psi\in S^{(n)}$ and $\Pi^{(n)}F^{(n)}\Pi^{(n)}=F^{(n)}$ as well as $0\leq F^{(n)}\leq\unity$.
		
		Now, fix any vector $\psi \in \complex^\infty$ and define the family $\varphi^{(n)} = F^{(n)} \psi \in S^{(n)}$. Let $m \leq n$, and note that every $\alpha\in S^{(m)}\subset S^{(n)}$ satisfies $\langle\alpha| F^{(m)} |\alpha\rangle = \f(\alpha) = \langle\alpha| F^{(n)} |\alpha\rangle$, and thus, by polarization, we also have that $\langle\alpha| F^{(m)} |\beta\rangle = \langle\alpha| F^{(n)} |\beta\rangle$ for all $\alpha, \beta\in S^{(m)}$.
		
		Now, define the sequences of complex numbers $x^{(n)}_j = \langle j|\varphi^{(n)}\rangle$ and $y_j = \langle j|\psi\rangle$. 
		For any $j\leq m$ we have
		\begin{align}
			\nonumber
			\left| x_j^{(n)}-x_j^{(m)} \right|^2
			&=  
			\left| \langle j|\varphi^{(n)}\rangle - \langle j|\varphi^{(m)}\rangle\right|^2 \\ \nonumber
			& =
			\left| \langle j| F^{(n)} |\psi\rangle - \langle j| F^{(m)} |\psi\rangle\right|^2   \\ \nonumber
& =
			\left| \langle j| F^{(n)}
			\mbox{$\sum_{i=m+1}^n$}\, y_i |i\rangle\right|^2 
			\\ \nonumber &\leq 
			\left\| \mbox{$\sum_{i=m+1}^n$}\, 
			y_i|i\rangle\right\|^2 
			=
			\mbox{$\sum_{i=m+1}^n$}\, |y_i|^2 
			\\ & \leq 
			\sum_{i=m+1}^\infty |y_i|^2 
			\ \xrightarrow{m\to\infty}\ 0
			\ .
		\end{align}
		Hence, for every $j$, the sequence $(x_j^{(n)})_{n\in\mathbb{N}}$ is a Cauchy sequence, which has some limit $x_j = \lim_{n\to\infty} x_j^{(n)}$.
		
		For all $N\in\mathbb{N}$ we have $\sum_{j=1}^N |x_j^{(n)}|^2 \leq \|\varphi^{(n)}\|^2 \leq \|\psi\|^2$, and thus $\sum_{j=1}^N |x_j|^2 = \lim_{n\to\infty} \sum_{j=1}^N |x_j^{(n)}|^2 \leq \|\psi\|^2$. 
		This implies that the object $\varphi = \sum_{j=1}^\infty x_j |j\rangle$ has finite norm $\|\varphi\|^2 = \sum_{j=1}^\infty |x_j|^2 \leq \|\psi\|^2$, and it is therefore a vector $\varphi \in \complex^\infty$.
		
		The above construction produces one output vector $\varphi\in \complex^\infty$ for each input vector $\psi \in \complex^\infty$. 
		This defines a map $F: \complex^\infty \to \complex^\infty$ via $F(\psi) = \varphi$. 
		Moreover, it is easy to check that $F(\lambda\psi) = \lambda F(\psi)$ for any $\lambda \in \complex$, and $F(\psi+\psi') = F(\psi) + F(\psi')$ for any $\psi, \psi' \in \complex^\infty$. Hence $F$ is a linear operator. Since $\|F(\psi)\|\leq \|\psi\|$ the operator $F$ is bounded and hence continuous.
		
		Suppose $\psi\in S$, then there exists some $n\in\mathbb{N}$ such that $\psi\in S^{(n)}$. By construction of $F$, for all $j\in\mathbb{N}$, we have
\begin{align*}
\langle j|F|\psi\rangle &  = \langle j|\varphi\rangle =x_j =\lim_{k\to\infty} x_j^{(k)} \\
			&  = \lim_{k\to\infty}\langle j|\varphi^{(k)}\rangle
			= \lim_{k\to\infty} \langle j|F^{(k)}|\psi\rangle.
\end{align*}
		In particular, if $1\leq j \leq n$, then $|j\rangle,\psi\in S^{(n)}$, and so $\langle j|F^{(k)}|\psi\rangle = \langle j|F^{(n)}|\psi\rangle$ for all $k\geq n$, thus $\langle j|F|\psi\rangle=\langle j|F^{(n)}|\psi\rangle$. We thus obtain
\begin{align*}
			\f(\psi) &= \langle \psi|F^{(n)}|\psi\rangle = \sum_{j=1}^n \bar y_j \langle j|F^{(n)}|\psi\rangle \\
			&=\sum_{j=1}^n \bar y_j \langle j|F|\psi\rangle =\langle \psi|F|\psi\rangle.
\end{align*}
		This proves \emph{existence} in our intermediate claim, now we would like to prove \emph{uniqueness}. To this end, suppose that both $F$ and $G$ are bounded operators such that $\f(\psi)=\langle \psi|F|\psi\rangle=\langle \psi|G|\psi\rangle$ for all normalized $\psi\in S$. Then the bounded operator $\Delta:=F-G$ satisfies $\langle\psi|\Delta|\psi\rangle=0$ for all $\psi\in S$. Since every vector in $\mathbb{C}^\infty$ can be approximated in norm to arbitary accuracy by elements in $S$, and since $\Delta$ is continuous, this shows that $\langle\psi|\Delta|\psi\rangle=0$ for all $\psi\in\complex^\infty$, and thus $\Delta=0$ since $\Delta$ is bounded and the Hilbert space is complex~\cite{Godlberg_on_1982}.
		
		This proves our intermediate claim. Since $\f(\psi)\in [0,1]$ for all normalized $\psi\in S$, and all normalized vectors in $\complex^\infty$ can be approximated in norm by normalized vectors in $S$, we have $\inf_\psi \langle\psi|F|\psi\rangle\geq 0$ and $\sup_\psi \langle \psi|F|\psi\rangle\leq 1$, where infimum and supremum are over all normalized vectors in $\complex^\infty$. Thus, $0\leq F\leq \unity$, and hence $F$ is self-adjoint.
		
		Let $\zeta\in\complex^\infty$ be an arbitrary normalized vector. If $\zeta\in S$ then, by construction, $\f(\zeta)=\langle\zeta|F|\zeta\rangle$. Now we want to show that this equation is also true if $\zeta\not\in S$. In this case, define the sequence of subspaces $T_1:={\rm span}\{\zeta\}$ and $T_{n+1}:={\rm span}\left(S_n\cup \{\zeta\}\right)$ for all $n\in\mathbb{N}$. Clearly $\dim T_n=n$ and $\bar T=\complex^\infty$ for $T=\bigcup_{n\in\mathbb{N}} T_n$. Thus, according to our intermediate claim, there is a bounded operator $G$ such that $\f(\psi)=\langle \psi|G|\psi\rangle$ for all normalized $\psi\in T$; in particular, $\f(\zeta)=\langle\zeta|G|\zeta\rangle$. But since $S\subset T$, we also have $\f(\psi)=\langle \psi|G|\psi\rangle$ for all $\psi\in S$. But, according to our intermediate statement, $F$ is the unique bounded operator satisfying this equation, hence $F=G$.
		
		As a side remark, note that the operator sequence $F^{(n)}$ does not in general converge to $F$ in operator norm.
	\end{proof}

\begin{Theorem}\label{theo:maininfty}
	Suppose that for each finite $d$ all OPFs $\f \in \F_d$ are of the form \eqref{eqFormOfOPFs}.
	Then the ``closedness under system composition" assumption (Definition~\ref{def:closedness}) implies that all OPFs $\f \in \F_\infty$ are also of the form \begin{equation}
		\f(\psi) = 
		\langle\psi| F|\psi\rangle\ ,
		\label{eqFormOfOPF}
	\end{equation}
	where the $\complex^\infty$-operator $F$ satisfies $0\leq F \leq \unity$.
\end{Theorem}


\begin{proof}
	Let us fix  a finite-dimensional subspace $S\subset \complex^\infty$.
	Denote the dimension of $S$ by $d$. 
	Let us fix an orthonormal basis $\psi_1, \ldots, \psi_d$ of $S$, an orthonormal basis $\alpha_1, \ldots, \alpha_d$ of $\complex^d$, and a normalized vector $\beta \in \complex^\infty$. 
	
	The Hilbert spaces $\complex^\infty$ and $\complex^d \otimes \complex^\infty$ are isomorphic in a very non-unique way; so let $X: \complex^d \otimes \complex^\infty \to \complex^\infty$ be an isometry such that
	\begin{equation}
		X (\alpha_i \otimes \beta)
		= \psi_i \ ,
	\end{equation}
	for all $i=1, \ldots, d$ (this does not determine $X$ uniquely; we will pick any such $X$ arbitrarily).
	Hence, for any vector $\psi \in S$ there is $\alpha \in \complex^d$ such that $\psi = X (\alpha \otimes \beta)$.
	And for any OPF $\f$ of $\complex^\infty$,the OPF $G:=f\circ X$ must be well-defined, since $\mathcal{F}_\infty$ is closed under composition with unitaries. In particular,
	\begin{equation}
		\label{eq:f to g}
		\f(\psi) 
		= \f(X(\alpha \otimes \beta))
		= \g(\alpha \otimes \beta) \qquad\mbox{for all }\psi\in S.
	\end{equation}
	Note that due to the mentioned isomorphism both, $\f$ and $\g$, belong to $\F_\infty$.
	
	At this point we invoke ``closedness under system composition" (Definition~\ref{def:closedness}).
	This tells us that for any OPF $\g \in \F_\infty$ of $\complex^d \otimes \complex^\infty$ there is $\mathbf h \in \F_d$ such that
	\begin{equation}
		\label{eq:h to g}
		{\bf h} (\alpha) 
		= \g (\alpha \otimes \beta)
		\ ,
	\end{equation}
	for all $\alpha \in \mathrm P \complex^d$.
	This together with Corollary \ref{Coro:main} implies that there is a $\complex^d$-matrix $H$ such that $0\leq H\leq \unity$ and $\mathbf h (\alpha) = \langle \alpha| H|\alpha \rangle$.
	
	Next we decompose $H$ in the chosen orthonormal basis of $\complex^d$, obtaining $H = \sum_{i,j=1}^d h_{ij} |\alpha_i\rangle\! \langle \alpha_j|$. 
	Also, we use the coefficients $h_{ij}$ to define the $(\complex^d \otimes \complex^\infty)$-operator $F_S = \sum_{i,j=1}^d h_{ij} |\psi_i\rangle\! \langle \psi_j|$, which is supported on the subspace $S$ and satisfies $0\leq F_S \leq \unity$. 
	
	Finally, for any given normalized $\psi \in S$, we decompose it in the chosen $S$-basis $\psi = \sum_{i=1}^d x_i\, \psi_i$; it follows that $\alpha = \sum_{i=1}^d x_i\, \alpha_i \in \complex^d$. 
	Combining this with \eqref{eq:f to g} and \eqref{eq:h to g} we obtain
	\begin{align}
		\f(\psi)
		&=
		g(\alpha\otimes\beta)
		= \mathbf h(\alpha)
		= \langle\alpha| H |\alpha\rangle  \nonumber \\ 
		& =
		\sum_{ik} \bar x_i h_{ij} x_j
		= 
		\langle\psi| F_S |\psi\rangle\ .
	\end{align}
	In summary, for any given finite-dimensional subspace $S$, we have constructed a $\complex^\infty$-operator $F_S$ satisfying the premises of Lemma~\ref{LemSubspace}.
	This gives us the conclusion of Theorem~\ref{theo:maininfty}.
\end{proof}

\section{The post-measurement state-update Rule}
\label{app:post-measurement}

\noindent
Until now we have been concerned with the outcome probabilities of quantum measurements.
In this section, we characterize the transformation that the quantum state undergoes during the measurement process.

\begin{Lemma}[quantum post-measurement state-update rule]
	\label{lem:post}
	The only post-measurement state-update rule compatible with the quantum probability assignment (\ref{eqFormOfOPFs}-\ref{eqFormOfstar}) is such that each measurement outcome is represented by a completely-positive linear map $\Lambda$. The probability of this outcome is given by
	\begin{equation}
		\label{eq:Lambda prob}
		P(\Lambda|\psi)
		=
		\tr \Lambda (\proj\psi)
		\ ,
	\end{equation}
	and the post-measurement state after outcome $\Lambda$ is
	\begin{equation}
		\label{eq:Lambda sigma}
		\rho
		=
		\frac {\Lambda (\proj\psi)}
		{\tr \Lambda (\proj\psi)}\ .
	\end{equation}
\end{Lemma}

In this statement each outcome is characterized by a map $\Lambda$, while in Corollary \ref{Coro:main} each outcome is characterized by s POVM elements $F$.
This two mathematical descriptions of an outcome are connected via
\begin{equation}
	\tr \Lambda (\proj\psi)
	=
	\langle\psi |F| \psi\rangle\ ,
\end{equation}
for all $\psi$.

The {remainder} of this section constitutes the proof of Lemma \ref{lem:post}. 
While the mathematics of this proof is certainly not new, we give the details in terms of the context and formalism of this paper.

\begin{proof}
	Corollary~\ref{Coro:main} states that any measurement has OPFs $\{\f_i \}$ of the form $\f_i (\psi) = \tr[F_i \proj\psi]$, where $\{F_i\}$ are positive operators  satisfying $\sum_i F_i =\unity$, that is, a POVM.
	This implies that all the statistical information of any ensemble $(\psi_r, p_r)$ is given by the corresponding density matrix $\rho = \sum_r p_r \proj{\psi_r}$.
	The associated linear form $\Omega_\rho :\complex\F_d \to \complex$ is given by $\Omega_\rho (\f_i) = \tr(\rho F_i)$, relating the usual density matrix formalism to the general formalism of this paper.
	
	At this point we still have not said anything about the post-measurement state update rule.
	But whatever this rule is, let $\sigma (F_i,\rho)$ be the post-measurement state (that is, its density matrix) after outcome $F_i$, when the initial state is $\rho$.
	And define the map $\Lambda_{F_i}$ which takes the original state $\rho$ to the post-measurement state times its corresponding probability:
	\begin{equation}
		\Lambda_{F_i} (\rho) := 
		\sigma (F_i,\rho)\, \tr[F_i \rho]\ .
	\end{equation}  
	
	Next, consider another given measurement with POVM $\{G_j\}$, and define the POVM $\{H_{j,i}\}$ to be that corresponding to the successive implementation of the measurements $\{F_i\}$ and $\{G_j\}$. 
	(This must correspond to a valid measurement, because the whole point of talking about a post-measurement state is that one can make further measurements on it.) 
	Then, using the rules of probability calculus and the above formulas we obtain
	\begin{eqnarray*}
		\tr[H_{j,i} \rho]
		&=&
		P(j,i)
		=
		P(j|i)P(i)=
		\tr[G_j \sigma(F_i,\rho)]\tr[F_i\rho]\\
		&=&
		\tr[G_j \Lambda_{F_i} (\rho)] \ ,
	\end{eqnarray*}
	for all $i,j$ and $\rho$.
	This equation implies that the map $\Lambda_{F_i} (\rho)$ is linear in $\rho$.
	
	By definition, the map $\Lambda_{F_i}$ takes every valid density matrix to a non-negative multiple of another valid density matrix, hence, the map $\Lambda_{F_i}$ is positive and trace-non-increasing. 
	To recover formulas \eqref{eq:Lambda prob} and \eqref{eq:Lambda sigma} we use the fact that $\tr\, \sigma (F_i, \rho) =1$, which gives $\tr \Lambda_{F_i} (\rho) = \tr[F_i\rho] =P(F_i|\rho)$.
	In summary, the probability of an outcome is the trace of the unnormalized post-measurement state given by the map $\Lambda$ associated to the outcome $F$ under consideration.
	This allows to fully characterize an outcome with the corresponding map $\Lambda$, with no reference to a POVM element $F$.
	
	Finally, we show that each outcome map $\Lambda$ is not just positive, but completely positive.
	As argued in the main text, we use the fact that one can always regard a system $\complex^d$ as part of a larger system $\complex^d \otimes \complex^b$. Then, the outcome map $\Lambda$ must remain a valid outcome map when extended to the larger system $\Lambda \otimes \mathcal I$, where $\mathcal I$ is the identity map on the Hermitian operators acting on $\complex^b$. 
	This is the definition of complete positivity.
\end{proof}

\section{Technical result}
\label{app:Technical results}

\begin{Lemma}
	\label{lemma:kernel irrep}
	The kernel of the partial-trace map 
	\begin{align}
		\label{contraction map I 2}
		\tr_n : \M_n^d &\to \M_{n-1}^d\ ,
		\\ 
		M &\mapsto \tr_n M\ ,
	\end{align}
	is the $\SU(d)$ irrep with Dynkin diagram 
	\begin{itemize}
		\item $(2n)$ if $d=2$ (also known as spin=$n$),
		
		\item $(n, \underbrace{0,\ldots,0}_{d-3} ,n)$ if $d\geq 3$.
	\end{itemize}
	We denote these family if irreps by $\D_n^d$.
\end{Lemma}

\begin{proof}
	First we note that the element $N_{n,n} = |0\rangle\! \langle 1|^{\otimes n} \in \M_n^d$ satisfies $\tr_n N_{n,n} =0$, so it is contained in the kernel of the map \eqref {contraction map I 2}.
	Also, we note that the element $N_{n,n}$ is the highest weight vector of the irrep $\D_n^d$ having the Dynkin diagram specified in the statement of this lemma.
	Hence the irrep $\D_n^d$ is contained in the kernel.
	
	To complete the proof of this lemma we only need to show that $\D_n^d$ is the only irrep inside the kernel.
	This is equivalent to the dimensional matching
	\begin{equation}
		\label{eq:dimensional matching}
		{\rm dim}\M_n^d =
		{\rm dim} \D_n^d +
		{\rm dim}\M_{n-1}^d\ ,
	\end{equation}
	implied by the Isomorphism Theorem.
	In order to check the above identity we use the dimensional formula given in page 224 of \cite{Fulton91}, which tells us
	\begin{align}\label{dimension}
		{\rm dim} \M_n^d
		&= 
		\binom{d+n-1}{n}^2 ,
		\\
		{\rm dim} \D_n^d 
		&=  
		\left(\frac{2n}{d-1}+1\right) 
		\prod_{k=1}^{d-2} 
		\left(1 +\frac n k \right)^2 .
	\end{align}
	With some calculation we get
	\begin{align}
		\nonumber
		& \hspace{-5mm}
		{\rm dim}\M_n^d - {\rm dim}\M_{n-1}^d
		\\ \nonumber &=
		\binom{d+n-1}{n}^2 - 
		\binom{d+n-2}{n-1}^2 
		\\ \nonumber &= 
		\frac{(d+n-1)!^2}{(d-1)!^2\, n!^2} - 
		\frac{(d+n-2)!^2}{(d-1)!^2 (n-1)!^2} 
		\\ \nonumber &=
		\frac{(d+n-1)!^2 - n^2(d+n-2)!^2}{(d-1)!^2\,  n!^2} 
		\\ \nonumber &=
		\frac{(d+n-2)!^2}{ n!^2 (d-2)!^2}\ 
		\frac{(d+n-1)^2-n^2}{(d-1)^2}
		\\ \label{eq:dim kernel} 
		&=
		\frac{(d+n-2)!^2}{ n!^2 (d-2)!^2}\ 
		\frac{d-1+2n}{d-1}\ ,
	\end{align}
	and
	\begin{align}
		\nonumber
		{\rm dim} \D_n^d 
		= &
		\left(\frac{2n}{d-1}+1\right) 
		\prod_{k=1}^{d-2} 
		\left(1 +\frac n k \right)^2 
		\\ \nonumber = & \frac{2n + d-1}{d-1} 
		\left[(1+n)
		\left(1+\frac{n}{2} \right) \cdots 
		\left(1+\frac{n}{d-2} \right) \right]^{\!2}
		\\ \nonumber = &
		\frac{2n + d-1}{d-1} 
		\left[ \frac{n+1} 1
		\frac{n+2}{2}  \cdots 
		\frac{n+d-2}{d-2} \right]^2
		\\ \label{eq:dim D} = &
		\frac{(d+n-2)!^2}{n!^2 (d-2)!^2}
		\frac{2n + d-1}{d-1}\ .
	\end{align}
	This shows that the dimensional matching \eqref{eq:dimensional matching} holds.
\end{proof}

\section{Other work}\label{app:Other_work}

In this section we compare the theorem presented in this work with recent work in the same direction. 
	
Ref.~\cite{Cabello_the_2018}  considers probability assignments (i.e. correlation tables) for sets of measurements and their exclusivity relations. It is shown that the \emph{exclusivity principle} (derived from properties of ideal measurements), together with an assumption of composability of experiments, restricts those correlations to be exactly those allowed by quantum theory. However, this does not prove that outcome probabilities of measurements on quantum states must be given by the Born rule; states and unitaries do not play any role in~\cite{Cabello_the_2018}. This is a very different approach from the one in our paper. 
We do not assume the existence of ideal measurements, but show that the Born rule follows (under minimal operational assumptions) from the dynamical postulates of quantum theory.
	
In~\cite{Frauchiger_nonprobabilistic_2017} the Born rule is recovered from postulates which are non-probabilistic and the assumption that measurement outcomes correspond to projectors. This is comparable to the decision theoretic approach of Deutsch~\cite{Deutsch_quantum_1999} and Wallace~\cite{Wallace_how_2010} which also seeks to account for the existence of probabilities. This is in contrast to the present work, where we do not seek to explain the emergence of probabilities in quantum theory, nor do we associate measurement outcomes to projectors.

\end{document}